\documentclass[10pt,conference]{IEEEtran}
\usepackage[top=0.75in,bottom=1.05in,left=0.625in,right=0.63in]{geometry}
\usepackage{graphicx}
\usepackage{multicol}
\usepackage{diagbox}
\usepackage{blkarray}

\usepackage{hyperref}
\usepackage{float}
\usepackage{url}
\usepackage[utf8]{inputenc} 
\usepackage{array}
\usepackage[utf8]{inputenc} 
\usepackage[T1]{fontenc}
\usepackage{mathtools}
\usepackage{amsfonts}
\usepackage{amsthm}
\usepackage{amsmath}
\usepackage{tabstackengine}
\usepackage{amssymb}
\usepackage{enumerate}
\usepackage{cite}
\usepackage{calc}
\usepackage{nicematrix}
\usepackage{color}
\usepackage{epsfig}
\usepackage{setspace}
\usepackage{pstricks}
\usepackage{cancel}
\usepackage{multirow}
\usepackage{mathrsfs}
\usepackage{algorithm}
\usepackage{algpseudocode}
\usepackage{romannum}
\usepackage{diagbox}
\algnewcommand{\algorithmicor}{\textbf{ or }}
\algnewcommand{\OR}{\algorithmicor}
\usepackage{booktabs,makecell}
\usepackage{mathrsfs}
\usepackage{enumitem}
\newtheorem{defn}{Definition}
\newtheorem{thm}{Theorem}
\newtheorem{cor}{Corollary}

\newtheorem{remark}{Remark}

\newtheorem{claim}{Claim}
\newtheorem{example}{Example}
\graphicspath{{./images/}}
\def\Ddots{\mathinner{\mkern1mu\raise\p@
		\vbox{\kern7\p@\hbox{.}}\mkern2mu
		\raise4\p@\hbox{.}\mkern2mu\raise7\p@\hbox{.}\mkern1mu}}
\makeatother

\begin{document}
	\pagenumbering{arabic}
	\title{Coded Caching for Combinatorial Multi-Access Hotplug Networks from $t$-Designs}%Combinatorial Multi-Access Coded Caching with Private Caches}

%\author{%
%	\IEEEauthorblockN{\IEEEauthorblockN{Dhruv Pratap Singh, Anjana A. Mahesh, \IEEEmembership{Member, IEEE} and B. Sundar Rajan, \IEEEmembership{Fellow, IEEE}}}
	%		\thanks{D. P. Singh, A. A. Mahesh and B. S. Rajan are with the Department of Electrical Communication Engineering, Indian Institute of Science, Bangalore,	560012, India (e-mail: \{dhruvpratap, anjanamahesh, bsrajan\}@iisc.ac.in).}}% <-this % stops a space
%\thanks{D. P. Singh and B. S. Rajan are with the Department of Electrical Communication Engineering, Indian Institute of Science, Bangalore 560012, India (e-mail: \{dhruvpratap, bsrajan\}@iisc.ac.in). A. A. Mahesh is with the Department of Electrical Engineering, Indian Institute of Technology Hyderabad, Telangana, 502285, India (e-mail: anjana.am@ee.iith.ac.in). A part of the manuscript appears in the proceedings of \textit{2025 IEEE Wireless Communications and Networking Conference (WCNC)}\cite{SMR}.}
%}
	\author{\IEEEauthorblockN{Dhruv Pratap Singh}
	\IEEEauthorblockA{
		\textit{Dept. of ECE, Indian Inst. of Sc.}\\
		Bangalore, KA, India \\
		dhruvpratap@iisc.ac.in}
	\and
	\IEEEauthorblockN{Anjana A. Mahesh}
	\IEEEauthorblockA{
		\textit{Dept. of EE, Indian Inst. of Tech.}\\
		Hyderabad, Telangana, India \\
		anjana.am@ee.iith.ac.in}
	\and
	\IEEEauthorblockN{B. Sundar Rajan}
	\IEEEauthorblockA{
		\textit{Dept. of ECE, Indian Inst. of Sc.}\\
		Bangalore, KA, India \\
		bsrajan@iisc.ac.in}}
\date{}
\maketitle
\begin{abstract}
We study hotplug coded caching in combinatorial multi-access networks, which generalizes existing hotplug coded caching models by allowing users to access multiple caches, while only a subset of caches is online during the delivery phase. We first generalize the Hotplug Placement Delivery Array (HpPDA) framework to the combinatorial multi-access setting. Based on this generalized framework, we propose a $t$-design-based coded caching scheme for combinatorial multi-access networks. We characterize a class of design parameters under which every active user has access to a sufficient number of coded subfiles to decode its requested file, and show that appropriate parameter choices allow for the elimination of redundant multicast transmissions. As a result, the proposed scheme achieves a family of rate–memory trade-offs with flexible subpacketization. We present numerical comparisons illustrating that the proposed $t$-scheme outperforms existing hotplug coded caching schemes in certain memory regimes \cite{MT,RR1,CRR}.
\end{abstract}
\section{Introduction}
To mitigate peak-time congestion in communication networks, caching has been proposed as a means to shift traffic from peak to off-peak periods by storing content closer to end users. In the coded caching framework introduced by Maddah-Ali and Niesen \cite{MAN}, cached content is further exploited to create multicast coding opportunities, thereby reducing the communication load. The coded caching system operates in two phases: a \textit{placement phase}, during which the server populates the caches without knowledge of future user requests, and a \textit{delivery phase}, during which the server transmits coded messages to satisfy the user demands. The seminal work in \cite{MAN} considers a centralized system with a server storing $N$ files and serving $K$ users, each equipped with a dedicated cache of size $M\leq N$ files. Under the assumption of uncoded cache placement, this scheme was shown to be optimal for $N\geq K$ in \cite{WTP}.

Following the landmark work on coded caching, the framework has been extended in several directions, including schemes with reduced subpacketization \cite{YCTC}, decentralized coded caching \cite{MAN1}, hierarchical coded caching \cite{KNMAD}, coded placement \cite{CFL}, multi-access coded caching \cite{MKR,HKD}, and coded caching with private demands \cite{WC}, among others. Among these, the combinatorial multi-access coded caching setting introduced in \cite{MKR} considers systems in which users access cache memories external to them, with each user connected to a fixed subset of caches. Under the assumption of uncoded placement, the achievable rate for this scheme was shown to be optimal in \cite{BE}.

Another important direction of the coded caching framework is the hotplug coded caching setting \cite{MT}, which relaxes the assumption that all users are active during the delivery phase. In practical systems, some users may be offline or fail to convey their demands at the time of delivery, whereas the original MAN scheme \cite{MAN} and subsequent works assume full user participation. In the hotplug model, the system consists of $K$ users, but only an unknown subset of $K^{\prime}$ users is active during the delivery phase. While the server is aware of the value of $K^{\prime}$ during the placement phase, it does not know the identities of the active users, and the cache placement must therefore be designed to satisfy any subset of $K^{\prime}$ users. The authors of \cite{MT} proposed a hotplug coded caching scheme based on the MAN framework \cite{MAN} using Maximum Distance Separable (MDS) coded placement. This line of work was further developed in \cite{RR1}, where a combinatorial structure known as the Hotplug Placement Delivery Array (HpPDA) was introduced and constructions based on $t$-designs were proposed. Further improvements in certain memory regimes were reported in \cite{CRR}.

Existing works on hotplug coded caching assume a dedicated caching architecture in which each user is associated with a single cache memory. In contrast, practical cache-aided networks often consist of multiple caches deployed across the network, where users may simultaneously access several caches. Moreover, due to failures, maintenance, or connectivity constraints, some caches may be unavailable during the delivery phase. 

Motivated by this, we study a hotplug coded caching model under a combinatorial multi-access framework. A server with $N$ files connects to $K$ users and $C$ caches via an error-free wireless link. Each user connects to a distinct $r$-subset of the $C$ caches, leading to $K=\binom{C}{r}$ users. Of the $C$ caches, an unknown subset of $C^{\prime}$ caches are online during the delivery phase. While the server is aware of the number of active caches, it does not know their identities at the time of placement, and the cache contents must therefore be designed to accommodate any such active subset. To the best of our knowledge, the combinatorial multi-access hotplug coded caching setting has not been previously investigated.

\textit{Our Contributions}: The main contributions of this paper are summarized as follows

\begin{itemize}
\item We formulate the combinatorial multi-access hotplug coded caching setting, in which users access a fixed subset of caches and only an unknown subset of caches is online during the delivery phase.

\item We generalize the Hotplug Placement Delivery Array (HpPDA) framework to the combinatorial multi-access setting by employing a structured family of placement and delivery arrays.

\item Using this generalized framework, we design a coded caching scheme with MDS-coded placement, leveraging $t$-designs, for the considered hotplug multi-access setting and characterize its achievable rate-memory trade-off.

\item We show that appropriate choices of design parameters allow the elimination of redundant multicast transmissions while preserving correct decoding for all active users.

\item We show that, for specific choices of system parameters, the proposed scheme specializes to the hotplug coded caching scheme in \cite{CRR}.

\item We present numerical comparisons illustrating that the proposed $t$-scheme outperforms existing hotplug coded caching schemes \cite{MT,RR1,CRR} in certain memory regimes.
\end{itemize}

\textit{Organization of the paper}: Section \ref{systemmodelandprelims} introduces the system model and relevant preliminaries. The main results are stated in Section \ref{mainresults}. Section \ref{proposedscheme} presents a generalized Hotplug Placement Delivery Array (HpPDA) framework and the associated placement and delivery procedure for combinatorial multi-access hotplug caching, while Section \ref{ghppdatdesign} details the combinatorial multi-access hotplug coded caching construction based on $t$-designs. Numerical comparisons are provided in Section \ref{numericalcomparisons}. Section \ref{conclusionsandfuturework} concludes the paper.

\textit{Notations}: For a positive integer $n$, the set $\{1,2,\cdots,n\}$ is denoted by $[n]$, and for positive integers $a,b$ such that $a<b$, the set $\{a,a+1,\cdots,b\}$ is denoted as $[a,b]$. For two sets $A$ and $B$, the notation $A\setminus B$ denotes the set of
elements in $A$ which are not in $B$. For a set $A$, the number of elements in it is represented by $|A|$. The binomial coefficient $\binom{n}{k}$ is equal to $\frac{n!}{k!(n-k)!}$, where $n$ and $k$ are positive integers such that $k \leq n$. For a set $S$ and a positive integer $t$, the notation $\binom{S}{t}$ denotes the set of all $t$-subsets of $S$.
\section{System Model and Preliminaries}
\label{systemmodelandprelims}
\subsection{System Model}
Consider a system consisting of a central server storing $N$ files of $B$ bits each, denoted by $\{W_1,W_2,\cdots,W_{N}\}$. The server is connected to the $K\leq N$ users via an error-free broadcast link. There are $C$ caches in the system, each with a memory of $M\leq N$ files or equivalently $MB$ bits. 

Users connect to these caches such that every $r$-subset of the $C$ caches is accessed by a distinct user, via an infinite capacity error-free wireless link. Consequently, there are $K=\binom{C}{r}$ users in the system, and each user can be uniquely indexed by the subset of caches it accesses. This system operates in two phases:

\textit{Placement Phase}: During the placement phase, the server populates the caches subject to their memory constraints, without knowledge of the future user demands or the identities of the caches that will be online during the delivery phase. However, the server is assumed to know that exactly $C^{\prime}$ caches will be online during the delivery phase. The contents of cache $a$ are denoted by $Z_a$ and the number of subfiles each file is partitioned into defines the subpacketization of the scheme.

\textit{Delivery Phase}: During the delivery phase, the set of online caches is denoted by $I\in\binom{[C]}{C^{\prime}}$, where $I$ is a $C^{\prime}$-subset of the $C$ caches. A user $U$ is active during delivery if $U\cap I\not=\emptyset$, that is, user $U$ connects to at least one online cache. Users satisfying $U\cap I=\emptyset$ are not served in this delivery phase and are treated separately, and hence are not considered part of the system during delivery. The demands of the active users are encapsulated in the demand vector $\mathbf{D}(I)=(d_U:U\in\binom{[C]}{r},U\cap I\not=\emptyset),$ for $d_U\in[N]$. On receiving the demand vector, the server broadcasts coded transmissions to satisfy the demands of the active users. The rate $R$ is defined as the number of bits transmitted normalized by the file size $B$, or equivalently, the number of transmissions normalized by the subpacketization level. The worst-case rate corresponds to the case where all active users request distinct files, resulting in the maximum number of transmissions. This model is referred to as the $(C,C^{\prime},r,N)$ Combinatorial Multi-Access Hotplug Coded Caching (CMAHCC) model.

The objective of the server is to jointly design the placement and delivery policies so as to minimize the worst-case rate.
\subsection{Preliminaries}
We first review the definition of Placement Delivery Arrays (PDAs) and then present the definition of Hotplug Placement Delivery Arrays (HpPDAs). Next, we introduce the definition of a $t$-design and summarize several results on $t$-designs that are used in this work. Finally, we briefly describe existing hotplug coded caching schemes that are used for comparison.

\begin{defn}(Placement Delivery Array \cite{YCTC})
	\label{pdadefn}
	For positive integers $K,F,Z$ and $S$, an $F\times K$ array $P=(p_{j,k})_{j\in[F],k\in[K]}$, composed of a specific symbol $\textasteriskcentered$ and $S$ non-negative integers $0,1,\cdots,S-1$, is called a $(K,F,Z,S)$ placement delivery array (PDA) if it satisfies the following conditions:
	\begin{itemize}
		\item[C1.] The symbol $\textasteriskcentered$ appears $Z$ times in each column;
		\item[C2.] Each integers occurs at least once in the array;
		\item[C3.] For any two distinct entries $p_{j_1,k_1}$ and $p_{j_2,k_2}$, $p_{j_1,k_1}=p_{j_2,k_2}=s$ is an integer only if
		\begin{itemize}
			\item[a.]$j_1\not=j_2,k_1\not=k_2$, i.e., they lie in distinct rows and distinct columns; and 
			\item[b.] $p_{j_1,k_2}=p_{j_2,k_1}=\textasteriskcentered$, i.e., the corresponding $2\times 2$ sub-array formed by the rows $j_1,j_2$ and columns $k_1,k_2$ must be of the following form
			 	\begingroup
			\setlength\arraycolsep{3pt}
			\begin{align*}
			\begin{bmatrix}
				\textasteriskcentered &s\\
				s &\textasteriskcentered
			\end{bmatrix} \text{ or }
						\begin{bmatrix}
							s &\textasteriskcentered\\
				\textasteriskcentered &s\\
				\end{bmatrix}.
		\end{align*}\endgroup
		\end{itemize}
	\end{itemize} 
\end{defn}
	A $[K,F,Z,S]$ PDA describes a dedicated coded caching scheme with $K$ users with cache memory $M$, $N$ files and subpacketization $F$ such that $\frac{M}{N}=\frac{Z}{F}$ and $R=\frac{S}{F}$.
	\begin{defn}(Hotplug PDA (HpPDA) \cite{RR1})
		Let $K, K^{\prime}, F, F^{\prime}, Z, Z^{\prime}$ and $S$ be integers such that $K \geq K^{\prime}, F\geq F^{\prime}$ and $Z<F^{\prime}$. Consider two arrays given as follows
		\begin{itemize}
		\item $P = (p_{f,k})_{f\in[F],k\in[K]}$ which is an array containing $\textasteriskcentered$	and null. Each column contains $Z$ number of $\textasteriskcentered$s.
		\item $B = (b_{f,k})_{f\in[F^{\prime}],k\in[K^{\prime}]}$ which is a $[K^{\prime}, F^{\prime}, Z^{\prime}, S]$-PDA.
	\end{itemize}
		For each $\tau\in \binom{[K]}{K^{\prime}}$, there exists a subset $\zeta\in\binom{[F]}{F^{\prime}}$ such that
		\begin{align*}
			[P]_{\zeta\times\tau}\overset{\textasteriskcentered}{=}B,
		\end{align*}where $[P]_{\zeta\times\tau}$ denotes the subarray of $P$ whose rows corresponds to the set $\zeta$ and columns corresponds to the set $\tau$, and $\overset{\textasteriskcentered}{=}$ denotes that the positions of all $\textasteriskcentered$ are same in both the arrays. We call it a $(K, K^{\prime}, F, F^{\prime}, Z, Z^{\prime}, S)$- HpPDA $(P, B)$.
	\end{defn}
	The following theorem gives the parameters of a hotplug	coded caching scheme obtained from an HpPDA.
	\begin{thm}(\cite{RR1})
		Given a $(K, K^{\prime}, F, F^{\prime}, Z, Z^{\prime}, S)$-HpPDA $(P, B)$, there exists a $(K, K^{\prime}, N)$ hotplug coded caching	scheme obtaining the following memory-rate pair,
		\begin{align*}
			\left(\frac{M}{N},R\right)=\left(\frac{Z}{F^{\prime}},\frac{S}{F^{\prime}}\right),
		\end{align*}where $M\leq N$ denotes the number of files each user can store in it’s cache.
	\end{thm}
	\begin{defn}(\cite{S})
		A design is a pair $(X, \mathcal{A})$ such that the following properties are satisfied:
		\begin{enumerate}
			\item $X$ is a set of elements called points, and
		\item $\mathcal{A}$ is a collection of nonempty subsets of $X$ called blocks.
	\end{enumerate}
	\end{defn}
\begin{defn}($t$-design\cite{S}). Let $v, k, \lambda$ and $t$ be positive integers such that $v>k\geq t$. A $t-(v, k, \lambda)$-design is a design $(X, \mathcal{A})$	such that the following properties are satisfied:
	\begin{enumerate}
		\item $|X| = v$,
	\item each block contains exactly $k$ points, and
	\item every set of $t$ distinct points is contained in exactly $\lambda$ blocks.
\end{enumerate}
\end{defn}
\begin{thm}(\cite{S}) \label{lambdas}Suppose that $(X, \mathcal{A})$ is a $t-(v, k, \lambda)$ design. Suppose that $Y \subset X$, where $|Y| = s \leq t$. Then there are exactly
	\begin{align*}
		\lambda_s=\frac{\lambda\binom{v-s}{t-s}}{\binom{k-s}{t-s}}
	\end{align*}blocks in $\mathcal{A}$ that contain all the points in $Y$.
\end{thm}
\begin{thm}(\cite{S}) Suppose that $(X, \mathcal{A})$ is a $t-(v, k, \lambda)$ design. Suppose that $Y, Z \subset X$, where $Y \cap Z = \emptyset, |Y| =i, |Z| = j$, and $i + j \leq t$. Then there are exactly
	\begin{align*}
		\lambda_i^{i+j}=\frac{\lambda\binom{v-i-j}{k-i}}{\binom{v-t}{k-t}}
	\end{align*}blocks in $\mathcal{A}$ that contain all the points in $Y$ and none of the points in $Z$.
\end{thm}The above theorem leads to the following corollary.
\begin{cor}(\cite{S}) \label{lambdats}Suppose that $(X, \mathcal{A})$ is a $t-(v, k, \lambda)$ design and	$Y \subset T \subset X$, where $|T| = t, |Y| = i$ with $i\leq t$. Then there are exactly
	\begin{align*}
		\lambda_s^t=\frac{\lambda\binom{v-t}{k-s}}{\binom{v-t}{k-t}}
	\end{align*}blocks in $\mathcal{A}$ that contain all the points in $Y$ and none of the points in $T\setminus Y$.
\end{cor}
\subsubsection{Some Existing Schemes on Hotplug Coded Caching}
We now present the MT scheme \cite{MT}, the CRR MT scheme and the CRR $t$-scheme \cite{CRR}, and the RR scheme \cite{RR1}. 

\textit{MT Scheme\cite{MT}}: Fix $t\in[K^{\prime}]$ and partition each file into $\binom{K^{\prime}}{t}$ subfiles as $W_{n}=\{W_{n,T}:T\in\binom{[K^{\prime}]}{t}\},\;\forall n \in[N]$. Then using an $\left[\binom{K}{t},\binom{K^{\prime}}{t}\right]$-MDS code, encode each subfile of every file as 
\begin{align*}
	\begin{bmatrix}
		C_{n,S^{\prime}_1}\\
		C_{n,S^{\prime}_2}\\
		\vdots\\
		C_{n,S^{\prime}_{\binom{K}{t}}}
	\end{bmatrix}=G^{T}\begin{bmatrix}
	W_{n,S^{\prime}_1}\\
	W_{n,S^{\prime}_2}\\
	\vdots\\
	W_{n,S^{\prime}_{\binom{K^{\prime}}{t}}}
	\end{bmatrix}, \forall n\in[N].
\end{align*}\textit{Placement Phase}: The cache contents of user $j\in[K]$ are:
\begin{align*}
	Z_j=\{C_{n,T}:T\in\binom{[K]}{t},j\in T,n\in[N]\}
\end{align*}\textit{Delivery Phase}: Let $I$ denote the set of online users during the delivery phase such that $I\in\binom{[K]}{K^{\prime}}$. The demands of the $K^{\prime}$ online users are given as $D[I]=\{d_{i_1},d_{i_2},\cdots,d_{i_{K^{\prime}}}\}$. Then for all $S\in\binom{[K^{\prime}]}{t+1}$, the server will broadcast the following transmission
\begin{align*}
	X_{S}=\sum\limits_{k\in S} C_{d_{k},S\setminus\{k\}}.
\end{align*}\textit{Rate Expression}: This scheme achieves the rate-memory trade-off given as 
\begin{align*}
	\left(M_t,R_t\right)=\left(N\frac{\binom{K-1}{t-1}}{\binom{K}{t}},\frac{\binom{K^{\prime}}{t+1}}{\binom{K^{\prime}}{t}}\right),\;\forall t\in[K^{\prime}].
\end{align*}\textit{CRR MT Scheme\cite{CRR}}: Fix $t\in[K^{\prime}]$ and partition each file into $\binom{K^{\prime}}{t}-\binom{K^{\prime}-1}{t-1}+\binom{K-1}{t-1}$ subfiles as $W_{n}=\{W_{n,T}:T\in\left[\binom{K^{\prime}}{t}-\binom{K^{\prime}-1}{t-1}+\binom{K-1}{t-1}\right]\},\;\forall n \in[N]$. Then using an $\left[\binom{K}{t},\left(\binom{K^{\prime}}{t}-\binom{K^{\prime}-1}{t-1}+\binom{K-1}{t-1}\right)\right]$-MDS code, encode each subfile of every file as 
\begin{align*}
\begin{bmatrix}
	C_{n,S^{\prime}_1}\\
	C_{n,S^{\prime}_2}\\
	\vdots\\
	C_{n,S^{\prime}_{\binom{K}{t}}}
\end{bmatrix}=G^{T}\begin{bmatrix}
	W_{n,S^{\prime}_1}\\
	W_{n,S^{\prime}_2}\\
	\vdots\\
	W_{n,S^{\prime}_{\left(\binom{K^{\prime}}{t}-\binom{K^{\prime}-1}{t-1}+\binom{K-1}{t-1}\right)}}
\end{bmatrix}, \forall n\in[N].
\end{align*}\textit{Placement Phase}: The cache contents of user $j\in[K]$ are:
\begin{align*}
Z_j=\{C_{n,T}:T\in\binom{[K]}{t},j\in T,n\in[N]\}
\end{align*}\textit{Delivery Phase}: Let $I$ denote the set of online users during the delivery phase such that $I\in\binom{[K]}{K^{\prime}}$. The demands of the $K^{\prime}$ online users are given as $D[I]=\{d_{i_1},d_{i_2},\cdots,d_{i_{K^{\prime}}}\}$. Then for all $S\in\binom{[K^{\prime}]}{t+1}$, the server will broadcast the following transmission
\begin{align*}
X_{S}=\sum\limits_{k\in S} C_{d_{k},S\setminus\{k\}}.
\end{align*}\textit{Rate Expression}: This scheme achieves the rate-memory trade-off given as 
\begin{align*}
\left(M_t,R_t\right)=&\left(\frac{N\binom{K-1}{t-1}}{\binom{K^{\prime}}{t}-\binom{K^{\prime}-1}{t-1}+\binom{K-1}{t-1}},\right.\\&\left.\frac{\binom{K^{\prime}}{t+1}}{\binom{K^{\prime}}{t}-\binom{K^{\prime}-1}{t-1}+\binom{K-1}{t-1}}\right),\;\forall t\in[K^{\prime}].
\end{align*}\textit{CRR $t$-Scheme \cite{CRR}}: Let $(X,\mathcal{A})$ a $t$-$(v,k,\lambda)$ design. An array $P$ whose columns are indexed by the points in $X$ and rows are indexed by the blocks in $\mathcal{A}$ is defined as
\begin{align*}
	P(A,i)=
	\begin{cases}
		\textasteriskcentered, &\text{ if }i\in A\\
		null, &\text{ if }i\not\in A
	\end{cases}.
\end{align*}For $0\leq a_s\leq \lambda^t_s,1\leq s\leq t-1$, consider a set 
\begin{align*}
	\mathcal{R}=\bigcup\limits_{s=1}^{t-1} \{(Y,i):Y\in\binom{[t]}{s},i\in[a_s]\}.
	\end{align*}Another array $B$ whose columns are indexed by $[t]$ and rows are indexed by the elements in $\mathcal{R}$ is defined as
	\begin{align*}
		B((Y,i),j)=
		\begin{cases}
			\textasteriskcentered, &\text{ if }j\in Y\\
			null, &\text{ if }j\not\in Y
		\end{cases}.
	\end{align*}The pair $(P, B)$ forms a $(K, K^{\prime}, F, F^{\prime}, Z, Z^{\prime}, S)$-HpPDA,	where $K=v,K^{\prime}=t,F=\lambda_0,F^{\prime}=\sum_{s=1}^{t-1}a_s\binom{t-1}{s-1}, Z=\lambda_1,Z^{\prime}=\sum_{s=1}^{t-1}a_s\binom{t-1}{s-1}$ and $S=\sum_{s=1}^{t-1}a_s\binom{t}{s+1}$.
	
	\textit{Placement Phase}: Each file is divided into $Z+F^{\prime}-Z^{\prime}$ subfiles as $W_{n}=\{W_{n,i}:i\in[Z+F^{\prime}-Z^{\prime}]\}$. Using an $\left[\lambda_0,Z+F^{\prime}-Z^{\prime}\right]$-MDS code, each subfile of every file is encoded and indexed by a block of the $t$-$(v,k,\lambda)$ design as $C_{n,A}:A\in\mathcal{A},\;\forall n\in[N]$. The contents of user $j\in[v]$ is given as
	\begin{align*}
		Z_j=\{C_{n,A}:A\in\mathcal{A},j\in A,\;\forall n\in[N]\}.
	\end{align*}\textit{Delivery Phase}: Let $I$ denote the set of online users during the delivery phase such that $I\in\binom{[v]}{t}$. The demands of these $t$ users are encapsulated in the demand vector $D[I]=\{d_{i_1},d_{i_2},\cdots,d_{i_t}\}$. The server makes a transmission for each non-star entry $s$ in $B$.
	
	\textit{Rate Expression}: The scheme achieves the following rate-memory trade-off 
	\begin{align*}
		(M,R)=\left(\frac{Z}{Z+F^{\prime}-Z^{\prime}},\frac{S}{Z+F^{\prime}-Z^{\prime}}\right).
	\end{align*}\textit{RR Scheme \cite{RR1}}: Let $(X,\mathcal{A})$ a $t$-$(v,k,\lambda)$ design. An array $P$ whose columns are indexed by the points in $X$ and rows are indexed by the blocks in $\mathcal{A}$ is defined as
	\begin{align*}
	P(A,i)=
	\begin{cases}
		\textasteriskcentered, &\text{ if }i\in A\\
		null, &\text{ if }i\not\in A
	\end{cases}.
	\end{align*}For $0\leq a_s\leq \lambda^t_s,1\leq s\leq t-1$, consider a set 
	\begin{align*}
	\mathcal{R}=\bigcup\limits_{s=1}^{t-1} \{(Y,i):Y\in\binom{[t]}{s},i\in[a_s]\},
	\end{align*}such that $|\mathcal{R}|>\lambda_1$. Another array $B$ whose columns are indexed by $[t]$ and rows are indexed by the elements in $\mathcal{R}$ is defined as
	\begin{align*}
	B((Y,i),j)=
	\begin{cases}
		\textasteriskcentered, &\text{ if }j\in Y\\
		null, &\text{ if }j\not\in Y
	\end{cases}.
	\end{align*}Further, define the PDA $\bar{B}$ which is obtained by removing all the integers in set	$T$ (defined in Theorem 10\cite{RR1}) from the array B. 
	
	\textit{Placement Phase}: Each file is divided into $|\mathcal{R}|$ subfiles as $W_{n}=\{W_{n,i}:i\in[|\mathcal{R}|]\}$. Using an $\left[\lambda_0,|\mathcal{R}|\right]$-MDS code, each subfile of every file is encoded and indexed by a block of the $t$-$(v,k,\lambda)$ design as $C_{n,A}:A\in\mathcal{A},\;\forall n\in[N]$. The contents of user $j\in[v]$ is given as
	\begin{align*}
		Z_j=\{C_{n,A}:A\in\mathcal{A},j\in A,\;\forall n\in[N]\}.
	\end{align*}\textit{Delivery Phase}: Let $I$ denote the set of online users during the delivery phase such that $I\in\binom{[v]}{t}$. The demands of these $t$ users are encapsulated in the demand vector $D[I]=\{d_{i_1},d_{i_2},\cdots,d_{i_t}\}$. The server makes a transmission for each non-star entry $s$ in $\bar{B}$.
	
	\textit{Rate Expression}: The scheme achieves the following rate-memory trade-off 
	\begin{align*}
		(M,R)=\left(\frac{\lambda_1}{\sum\limits_{s=1}^{t-1}a_s\binom{t}{s}},\frac{\sum\limits_{s=1}^{t-1}a_s\binom{t}{s+1}-|T|}{\sum\limits_{s=1}^{t-1}a_s\binom{t}{s}}\right).
	\end{align*}
	\section{Main Results}
	\label{mainresults}
	In this section, we present the main results of this work. In the following theorem, we characterize a set of achievable rate-memory points for the combinatorial multi-access hotplug system, using a $t$-$(v,k,\lambda)$ design, where the $v$ points represent the $v$ caches, each block of size $k$ indexes a coded subfile, and $t$ determines the number of online caches during the delivery phase.%give the rate-memory trade-off achieved by generalized HpPDA constructed using $t$-designs in Section \ref{ghppdatdesign}.
		\begin{thm}[$t$-scheme]Consider a $t$-$(v,k,\lambda)$ design and the corresponding combinatorial multi-access hotplug coded caching model with $C=v$ caches, $C^{\prime}=t$ online caches, access degree $r$, and library size $N$. For each $j\in[r]$ and $s\in[t-j]$, let the parameters $0\leq a_{s,j}\leq \lambda^t_s$ define \begin{align*}Y_j={\sum\limits_{i=1}^{j}(-1)^{i+1}\binom{j}{i}\lambda_i + \sum\limits_{s=1}^{t-j} a_{s,j}\binom{t-j}{s}}.\end{align*}If the parameters are chosen such that $Y_j\leq Y_{j-1}$ for all $j\in[r]$, then there exists a $(v,t,r,N)$ combinatorial multi-access hotplug scheme with the rate-memory trade-off
			\begin{align}
				\left(\frac{M}{N},R\right)&=\left(\frac{\lambda_1}{\sum\limits_{i=1}^{r}(-1)^{i+1}\binom{r}{i}\lambda_i + \sum\limits_{s=1}^{t-r} a_{s,r}\binom{t-r}{s}}\right.,\nonumber\\&\left.\frac{\sum\limits_{j=1}^{r}\sum\limits_{s=1}^{t-j}a_{s,j}\binom{t}{s+j}\binom{v-t}{r-j}}{\sum\limits_{i=1}^{r}(-1)^{i+1}\binom{r}{i}\lambda_i + \sum\limits_{s=1}^{t-r} a_{s,r}\binom{t-r}{s}}\right)
			\end{align}
		\end{thm}
	\begin{proof}
		The proof is presented in Section \ref{ghppdatdesign}.
	\end{proof}The theorem characterizes a family of achievable rate-memory pairs for the combinatorial multi-access hotplug coded caching model based on $t$-designs. The condition $Y_j\leq Y_{j-1}$ ensures that every active user has access to a sufficient number of coded subfiles to decode its requested file, while allowing users connecting to fewer online caches to receive additional coded subfiles from multicast transmissions. This flexibility enables a reduction in the overall number of server transmissions by allowing the parameters $a_{s,j}$ to be chosen strictly smaller than $\lambda^t_s$, thereby eliminating multicast transmissions that are unnecessary for satisfying the decodability requirement. Moreover, different choices of the parameters $a_{s,r}, s\in[t-r]$ lead to different subpacketization levels, which in turn result in different cache memory requirements and induce distinct cache placement patterns. Consequently, the proposed framework admits multiple feasible placements that achieve different points on the rate-memory trade-off.
	\begin{remark}
		Here, $Y_j$ denotes the total number of coded subfiles accessible to a user connecting to exactly $j$ online caches after adding the subfiles obtained through cache access and those obtained from server transmissions. The choice $a_{s,j}=\lambda^t_s$ for all admissible $s$ and $j$ ensures that every active user has access to exactly the same number of coded subfiles, equal to the minimum required for decoding, as shown in Claim \ref{as=lambdatsclaim}. For general choices with $a_{s,j}<\lambda^t_s$, the condition $Y_j\leq Y_{j-1}$ implies that users connecting to fewer online caches may obtain more coded subfiles than the decoding threshold. These additional coded subfiles correspond to multicast transmissions that are redundant for file recovery and can therefore be eliminated without affecting decodability. Consequently, appropriately selecting the parameters $a_{s,j}$ enables a reduction in the overall number of server transmissions while preserving correctness.
	\end{remark}
	\section{Proposed Hotplug Coded Caching Scheme for Combinatorial Multi-Access Networks}
	\label{proposedscheme}
We begin this section by extending the HpPDA framework to the combinatorial multi-access setting, where users may access different numbers of online caches. To capture this heterogeneity, we associate a PDA $B_j$ with users that access exactly $j$ online caches, for $j\in[r]$. We first formalize the generalized HpPDA with a single global placement array $P_c$, a derived user access array $P$, and a family of PDAs $B_j,\;\forall j\in[r]$. We then present an algorithm that specifies the corresponding placement and delivery procedures.
\subsection{Generalized Multi-Access Hotplug Placement Delivery Arrays}
\begin{defn}\label{CHpPDAdefn}Let $C,C^{\prime},r,F,Z_c,Z$ be integers with $C>C^{\prime}$, and let $F^{\prime}_j,Z^{\prime}_j,S_j$ be positive integers with $F>F^{\prime}_j$ for all $j\in[r]$. Consider the following arrays:
	\begin{itemize}
		\item $P_c=(\alpha_{f,i})_{f\in[F],i\in[C]}$ is an array composed of the symbol $\textasteriskcentered$ and null entries. Each column of $P_c$ contains exactly $Z_c$ instances of $\textasteriskcentered$.
		\item $P=(\beta_{f,U})_{f\in[F],U\in\binom{[C]}{r}}$, where $\beta_{f,U}=\textasteriskcentered$ if $\alpha_{f,c}=\textasteriskcentered$ and $c\in U$, is an array composed of $\textasteriskcentered$ and null entries. Each column of $P$ contains exactly $Z$ $\textasteriskcentered$ entries.
		
		\item $B_j=(b_{f,k})_{f\in[F^{\prime}_j],k\in\left[\binom{C^{\prime}}{j}\binom{C-C^{\prime}}{r-j}\right]}$, which is a $[\binom{C^{\prime}}{j}\binom{C-C^{\prime}}{r-j},F^{\prime}_j,Z^{\prime}_j,S_j]$-PDA, for each $j\in[r]$.
	\end{itemize}

For each $j\in[r]$ and every subset $I\in\binom{[C]}{C^{\prime}}$, let $\tau_j(I)=\{U\in\binom{[C]}{r}:|U\cap I|=j\}$. There exists a subset $\zeta_j(I)\in\binom{[F]}{F^{\prime}_j}$ such that 
\begin{align}\label{star=condition}
	[P]_{\zeta_j(I)\times\tau_j(I)}\overset{\textasteriskcentered}{=}B_j,
\end{align}where $[P]_{\zeta_j(I)\times\tau_j(I)}$ denotes the subarray of $P$ whose rows correspond to the set $\zeta_j(I)$ and whose columns correspond to the set $\tau_j(I)$, and $\overset{\textasteriskcentered}{=}$ denotes that the positions of all $\textasteriskcentered$ symbols coincide in both arrays. We refer to the collection $(P_c,P,\{B_j\}_{j=1}^{r})$ as a $(C,C^{\prime},r,F,Z_c,Z,\{F^{\prime}_j,Z^{\prime}_j,S_j\}_{j=1}^{r})$-generalized HpPDA.
\end{defn}
\begin{remark}
	When $r=1$, each user connects to a single cache. For any realization of the online cache set $I$, all active users satisfy $|I\cap U|=1$, and hence only the PDA $B_1$ exists. Further, for $r=1$, the arrays $P$ and $P_c$ coincide. Thus, Definition \ref{CHpPDAdefn} specializes to the standard HpPDA definition presented in \cite{RR1}.
\end{remark}We use the following example to illustrate Definition \ref{CHpPDAdefn}.
\begin{example}
	For $C=5,C^{\prime}=3,r=2,F=5,F^{\prime}_1=3,F^{\prime}_2=3,Z=2,Z^{\prime}_1=1,Z^{\prime}_2=2,S_1=1,S_2=6$, we have 
	 	\begingroup
	\setlength\arraycolsep{3pt}
	\begin{align*}
		&P_c=
		\begin{bmatrix}
			\textasteriskcentered & & & & \\
			 &\textasteriskcentered & & & \\
			 & &\textasteriskcentered & & \\
			 & & &\textasteriskcentered & \\
			 & & & &\textasteriskcentered 
		\end{bmatrix},
P=\begin{bmatrix}
 \textasteriskcentered &\textasteriskcentered &\textasteriskcentered &\textasteriskcentered & & & & & &\\
 \textasteriskcentered & & & &\textasteriskcentered &\textasteriskcentered &\textasteriskcentered & & &\\
 &\textasteriskcentered & & &\textasteriskcentered & & &\textasteriskcentered &\textasteriskcentered &\\
 & &\textasteriskcentered & & &\textasteriskcentered & &\textasteriskcentered & &\textasteriskcentered\\
 & & &\textasteriskcentered & & &\textasteriskcentered & &\textasteriskcentered &\textasteriskcentered
\end{bmatrix},\\
	 &B_1=
	\begin{bmatrix}
 \textasteriskcentered &\textasteriskcentered &1 &2 &3 &4 \\
 1 &2 &\textasteriskcentered &\textasteriskcentered &5 &6 \\
 3 &4 & 5 &6 &\textasteriskcentered &\textasteriskcentered 
\end{bmatrix}\text{and }
		B_2=\begin{bmatrix}
			 \textasteriskcentered &\textasteriskcentered &1 \\
			 \textasteriskcentered &1 &\textasteriskcentered \\
			 1 &\textasteriskcentered &\textasteriskcentered 
		\end{bmatrix}.
	\end{align*}
	\endgroup 
	For every $I\in\binom{[C]}{C^{\prime}}$, there exist $(\tau_1(I),\zeta_1(I))$ and $(\tau_2(I),\zeta_2(I))$ such that $[P]_{\zeta_j(I)\times\tau_j(I)}\overset{\textasteriskcentered}{=}B_j$. For instance, let $I=\{1,3,5\}$. In this case, we have $\tau_2(I)=\{13,15,35\}$ and $\zeta_2(I)=\{1,3,5\}$. Next, we have $\tau_1(I)=\{12,14,23,34,25,45\}$ and corresponding $\zeta_1(I)=\{1,3,5\}$.
\end{example}
\subsection{Hotplug Coded Caching Scheme for Combinatorial Multi-Access Networks}
\begin{algorithm}
	\caption{Combinatorial Multi-Access Hotplug Coded Caching Scheme using the $(C,C^{\prime},r,F,Z_c,Z,\{F^{\prime}_j,Z^{\prime}_j,S_j\}_{j=1}^{r})$-generalized HpPDA}
	\label{Algo1}
	\begin{algorithmic}[1]
		\Function{Placement Phase}{$P_c,W_n,\;n\in[N],G$} 
		\State \label{subfiledivide}Divide each file $W_n,\;n\in[N]$ into $F^{\prime}_r-Z^{\prime}_r+Z$ subfiles as $W_{n}=\{W_{n,i}:i\in[F^{\prime}_r-Z^{\prime}_r+Z]\},\;\forall n\in[N].$
		\State \label{subfilencode}Encode the $F^{\prime}_r-Z^{\prime}_r+Z$ subfiles of each file into $F$ coded subfiles using an $[F,F^{\prime}_r-Z^{\prime}_r+Z]$-MDS code with the generator matrix $G$ of dimensions $F^{\prime}_r-Z^{\prime}_r+Z\times F$ as
			\begingroup
		\setlength\arraycolsep{3pt}
		\begin{align*}
			\begin{bmatrix}
				W^{c}_{n,1}\\
				W^{c}_{n,2}\\
				\vdots\\
				W^{c}_{n,F}\\
			\end{bmatrix}=G^T\begin{bmatrix}
				W_{n,1}\\
				W_{n,2}\\
				\vdots\\
				W_{n,F^{\prime}_r-Z^{\prime}_r+Z}
			\end{bmatrix}
		\end{align*}\endgroup
		\For{$i\in[C]$}
		\State \label{subfileplacement}$Z_i\gets\{W^{c}_{n,f}:\alpha_{f,i}=\textasteriskcentered,n\in[N],f\in[F]\}$
		\EndFor
		\EndFunction
		\Function{Delivery Phase}{$\{B_j\}_{j=1}^{r},P,W_n,n\in[N],I,\mathbf{D}(I)$}
		\State \label{onlinecacheset}Let $I$ be the set of online caches with the demand vector being $\mathbf{D}(I)=(d_U:U\in\binom{[C]}{r},U\cap I\not=\emptyset)$.
		\State \label{star=defn}By Definition \ref{CHpPDAdefn}, for $(P,\{B_j\}_{j=1}^{r})$, there exists $(\tau_j(I),\zeta_j(I))$ such that $[P]_{\zeta_j(I)\times\tau_j(I)}\overset{\textasteriskcentered}{=}B_j,\forall j \in[r]$.
		\For{$j\in[r]$}
		\State \label{deliveryP}Make the array $\overline{P}=(\overline{p}_{f,U})_{f\in\zeta_j(I),U\in\tau_j(I)}$ by filling $s\in[S_j]$ integers in the null spaces of the subarray $[P]_{\zeta_j(I)\times\tau_j(I)}$ such that $\overline{P}=B_j$.
		\For{$s\in[S_j]$}
		\State \label{codedtransmissions}The server makes the transmissions \begin{align*}
			X_s=\bigoplus\limits_{\overline{p}_{f,U}=s,f\in\zeta_j(I),U\in\tau_j(I)} W^{c}_{d_U,f}.\end{align*}
		\EndFor
		\EndFor 
		\EndFunction
	\end{algorithmic}
\end{algorithm}
We now describe how a combinatorial multi-access hotplug coded caching scheme is constructed using the $(C,C^{\prime},r,F,Z_c,Z,\{F^{\prime}_j,Z^{\prime}_j,S_j\}_{j=1}^{r})$-generalized HpPDA. 

The placement and delivery phases are explained using Algorithm \ref{Algo1}. 

\textit{Placement Phase}: As shown in Line \ref{subfiledivide}, each file $W_n,n\in[N]$ is divided into $F^{\prime}_r-Z^{\prime}_r+Z$ subfiles. These subfiles are then encoded into $F$ coded subfiles using an $[F,F^{\prime}_r-Z^{\prime}_r+Z]$-MDS code as described in Line \ref{subfilencode}. The encoded subfiles of file $W_n$ are denoted by $W^{c}_{n,f},f\in[F]$. 

The array $P_c$ governs the cache placement. It is an $F\times C$ array whose columns correspond to caches with rows corresponding to the coded subfile indices. The cache $k$ stores the coded subfile $W^{c}_{n,f},n\in[N]$ if $\alpha_{f,k}=\textasteriskcentered$. 

Since there are $Z_c$ stars in each column of $P_c$, each cache stores $NZ_c$ coded subfiles or equivalently, $N\frac{Z_c}{F^{\prime}_r-Z^{\prime}_r+Z}$ files. Thus, we have \begin{align*}\frac{M}{N}=\frac{Z_c}{F^{\prime}_r-Z^{\prime}_r+Z}.\end{align*}\textit{Delivery Phase}: During the delivery phase, the set of online caches is denoted by $I$ and the demand vector of the active users is given by $\mathbf{D}(I)$ as specified in Line \ref{onlinecacheset}. 

By Definition \ref{CHpPDAdefn}, for each $j\in[r]$ and for a given $I$, there exist sets $\tau_j(I)=\{U\in\binom{[C]}{r}:|U\cap I|=j\}$ and $\zeta_j(I)$ such that $[P]_{\zeta_j(I)\times \tau_j(I)}\overset{\textasteriskcentered}{=}B_j$, as invoked in Line \ref{star=defn}.

The array $P$ represents the coded subfiles available to the users from the caches they connect to, where $\beta_{f,U}=\textasteriskcentered$ indicates that the coded subfile $W^c_{n,f},n\in[N]$ is available to user $U$. Using the PDA $B_j$, Algorithm \ref{Algo1} constructs the PDA $\overline{P}$ as given in Line \ref{deliveryP}, whose non-star entries correspond to multicast transmissions. Since $\overline{P}$ has $Z^{\prime}_j$ stars per column, each user accessing $j$ online caches receives $F^{\prime}_j-Z^{\prime}_j$ coded subfiles from server transmissions. 
 
For a given $I$, users are grouped according to the number of online caches they access. A user $U$ satisfying $|U\cap I|=j$ has access to the contents of $j$ online caches, for $j\in[r]$. Let $Z_j$ denote the total number of coded subfiles of each file available to such a user from the online caches it connects to. Since the delivery for these users is governed by the PDA $B_j$, which has $Z^{\prime}_j$ $\textasteriskcentered$ entries in each column, we have $Z_j\geq Z^{\prime}_j$. This process is repeated for each $j\in[r]$. 

Adding the coded subfiles obtained from cache access and server transmissions, a user that connects to exactly $j$ online caches has access to a total of $F^{\prime}_j-Z^{\prime}_j+Z_j$ coded subfiles. The delivery phase is designed such that
\begin{align*}
	F^{\prime}_j-Z^{\prime}_j+Z_j\geq F^{\prime}_r-Z^{\prime}_r+Z
\end{align*}for all $j\in[r]$. Consequently, every active user has access to at least $F^{\prime}_r-Z^{\prime}_r+Z$ coded subfiles, which is sufficient to reconstruct its requested file using the $[F,F^{\prime}_r-Z^{\prime}_r+Z]$-MDS code.

\textit{Rate Expression}: For each $j\in[r]$, the PDA $B_j$ generates $S_j$ multicast transmissions. Thus, the worst-case rate-memory trade-off is given by
\begin{align*}
	\left(\frac{M}{N},R\right)=\left(\frac{Z_c}{Z+F^{\prime}_r-Z^{\prime}_r},\sum\limits_{j=1}^{r}\frac{S_j}{Z+F^{\prime}_r-Z^{\prime}_r}\right)
\end{align*}
\begin{remark}
	Observe that for $r=1$, we have $Z_1=Z$, which implies that the above rate expression and Algorithm \ref{Algo1} reduce to the rate and Algorithm presented in \cite{CRR}.
\end{remark}

\section{Combinatorial Multi-Access Hotplug Coded Caching Using $t$-Designs}
\label{ghppdatdesign}
 In this section, we present how a combinatorial multi-access hotplug coded caching scheme is constructed using $t$-designs. In particular, we describe the construction of the cache placement array $P_c$, the user access array $P$ and the associated family of PDAs $B_j,j\in[r]$, and show how they jointly determine the placement and delivery procedures. Let $(X,\mathcal{A})$ be a $t$-$(v,k,\lambda)$ design, where $X=[v]$ and $\mathcal{A}=\{A_1,A_2,\cdots,A_b\}$. From Theorem \ref{lambdas}, we know that $b=\lambda_0=\frac{\lambda\binom{v}{t}}{\binom{k}{t}}$. 
 
 Consider the array $P_c=(P_c(A,i))_{A\in\mathcal{A},i\in[v]}$ whose rows are indexed by the blocks in $\mathcal{A}$ and columns by the points in $X$. The array $P_c$ is defined as
 \begin{align}\label{pcdefn}
 	P_c(A,i)=\begin{cases}
 		\textasteriskcentered, &\text{if }i\in A\\
 		 null, &\text{if }i\not\in A
 	\end{cases}.
 \end{align}Next, we define the array $P=(P(A,U))_{A\in\mathcal{A},U\in\binom{X}{r}}$ as
 \begin{align}\label{pdefn}
 	P(A,U)=\begin{cases}
 		\textasteriskcentered, &\text{if }U\cap A\not=\emptyset\\
 		null, &\text{if }U\cap A=\emptyset
 	\end{cases}.
 \end{align}Next, we define the PDA $B_j$. Consider the set $\mathcal{R}_j=\bigcup\limits_{s=1}^{t-j} \{(Y,i):Y\in\binom{[t]}{s},i\in[a_{s,j}]\}$, for $j\in[r],s\in[t-j]$, and $0\leq a_{s,j}\leq \lambda^t_s$, with $\lambda^t_s$ defined in Corollary \ref{lambdats}. Note that $|\mathcal{R}_j|=\sum\limits_{s=1}^{t-j} a_{s,j}\binom{t}{s}$. Consider an array $B_j$ whose rows are indexed by the elements of $\mathcal{R}_j$ and whose columns are indexed by the set $\mathcal{C}_j=\{U:U\in\binom{X}{r},|U\cap[t]|=j\}$. Hence, the array $B_j=(B_j((Y,i),U))$ is a $|\mathcal{R}_j|\times|\mathcal{C}_j|=\left(\sum\limits_{s=1}^{t-j}a_{s,j}\binom{t}{s}\times\binom{t}{j}\binom{v-t}{r-j}\right)$ array defined as
 \begin{align}
 	\label{Bjdefn}
 	B_j(&(Y,i),U)=\\
 	&\begin{cases}
 		\textasteriskcentered,& \text{if }\{U\cap[t]\}\cap Y\not=\emptyset\\
 		(Y\cup \{U\cap[t]\},i)_{n_{Y\cup \{U\cap[t]\}}},& \text{if }\{U\cap[t]\}\cap Y=\emptyset\nonumber
 	\end{cases},
 \end{align}where $n_{Y\cup \{U\cap[t]\}}$ denotes the occurrence order of the set $Y\cup \{U\cap[t]\}$ from left to right in $B_j$. We now present the following example to illustrate the working of the arrays defined above.
 \begin{example}
 	Consider the $3$-$(8,4,1)$ design $(X,\mathcal{A})$ where $X=[8]$ and $\mathcal{A}=\{1256, 3478, 2468, 1357, 1458, 2367, 1234, 5678, 1278, 3456,\\ 1368, 2457, 1467, 2358\}$. Further, let each user connect to $r=2$ out of the $|X|=8$ caches. The array $P_c$ is given as
 	\begin{align*}
 	\begingroup
\setlength\arraycolsep{3pt}
 	P_c=\begin{bNiceMatrix}[first-col,first-row]
 	 {}	 &1 &2 &3 &4 &5 &6 &7 &8 \\
 	 1256	&\textasteriskcentered &\textasteriskcentered & & &\textasteriskcentered &\textasteriskcentered & & \\
 	 3478	& & &\textasteriskcentered &\textasteriskcentered & & &\textasteriskcentered &\textasteriskcentered \\
 	 2468	& &\textasteriskcentered & &\textasteriskcentered & &\textasteriskcentered & &\textasteriskcentered \\
 	 1357	&\textasteriskcentered & &\textasteriskcentered & &\textasteriskcentered & &\textasteriskcentered & \\
 	 1458	&\textasteriskcentered & & &\textasteriskcentered &\textasteriskcentered & & &\textasteriskcentered \\
 	 2367	& &\textasteriskcentered &\textasteriskcentered & & &\textasteriskcentered &\textasteriskcentered & \\
 	 1234	&\textasteriskcentered &\textasteriskcentered &\textasteriskcentered &\textasteriskcentered & & & & \\
 	 5678	& & & & &\textasteriskcentered &\textasteriskcentered &\textasteriskcentered &\textasteriskcentered \\
 	 1278	&\textasteriskcentered &\textasteriskcentered & & & & &\textasteriskcentered &\textasteriskcentered \\
 	 3456	& & &\textasteriskcentered &\textasteriskcentered &\textasteriskcentered &\textasteriskcentered & & \\
 	 1368	&\textasteriskcentered & &\textasteriskcentered & & &\textasteriskcentered & &\textasteriskcentered \\
 	 2457	& &\textasteriskcentered & &\textasteriskcentered &\textasteriskcentered & &\textasteriskcentered & \\
 	 1467	&\textasteriskcentered & & &\textasteriskcentered & &\textasteriskcentered &\textasteriskcentered & \\
 	 2358	& &\textasteriskcentered &\textasteriskcentered & &\textasteriskcentered & & &\textasteriskcentered 
 	\end{bNiceMatrix}
 	 \endgroup
 	\end{align*}Since there are $b=14$ blocks and $v=8$ caches, the array $P_c$ has dimension $14\times8$. Due to space limitations, we do not display the full array $P$, which has dimensions $b\times\binom{v}{r}=14\times\binom{8}{2}=14\times28$. Instead, for a given online cache set $I$, we explicitly describe the subsets $\zeta_j(I)$ and $\tau_j(I)$, for each $j\in[r]$, and show that $[P]_{\zeta_j(I)\times\tau_j(I)}\overset{\textasteriskcentered}{=}B_j$.

 	For this example, from Corollary \ref{lambdats}, we have $\lambda^t_2=\lambda^3_2=2$ and $\lambda^t_1=\lambda^3_1=2$. Choosing $a_{1,2}=1, a_{1,1}=2$ and $a_{2,1}=1$, we obtain $\mathcal{R}_1=\{(12,1),(13,1),(23,1),(1,1),(2,1),(3,1),(1,2),(2,2),\\(3,2)\}$ with $|\mathcal{R}_1|=\sum\limits_{s=1}^{2}a_{s,1}\binom{3}{s}=9$ and $\mathcal{R}_2=\{(1,1),(2,1),(3,1)\}$ with $|\mathcal{R}_2|=\sum\limits_{s=1}^{1}a_{s,2}\binom{3}{s}=3$. Corresponding to $\mathcal{R}_1$, the column index set is $\mathcal{C}_1=\{14,15,16,17,18,24,25,26,27,28,34,35,36,37,38\}$ and for $\mathcal{R}_2$, we have $\mathcal{C}_2=\{12,13,23\}$. 
 	
 	 Due to page-width constraints, we display the PDA $B_1$ in two parts. Specifically, we write $B_1=[B_{1,1}\;\;B_{1,2}]$, where $B_{1,1}$ contains the first $8$ columns of $B_1$ and $B_{1,2}$ contains the remaining 7 columns as shown in Fig. \ref{B1}. This partition is purely for presentation purposes. The PDA $B_2$ is given below as 
 \begin{figure*}[t]
 	\begingroup
 	\setlength\arraycolsep{3pt}
 	\begin{align*}
 		%		\begingroup
 		\setlength\arraycolsep{3pt}
 		&B_{1,1}=\begin{bNiceMatrix}[first-col,first-row]
 			{} &14 &15 &16&17&18&24&25&26\\
 			(12,1) &\textasteriskcentered &\textasteriskcentered &\textasteriskcentered &\textasteriskcentered &\textasteriskcentered &\textasteriskcentered &\textasteriskcentered &\textasteriskcentered \\
 			(13,1) &\textasteriskcentered &\textasteriskcentered &\textasteriskcentered &\textasteriskcentered &\textasteriskcentered &(123,1)_1 &(123,1)_2 &(123,1)_3 \\
 			(23,1) &(123,1)_1 &(123,1)_2 &(123,1)_3 &(123,1)_4 &(123,1)_5 &\textasteriskcentered &\textasteriskcentered &\textasteriskcentered \\
 			(1,1) &\textasteriskcentered &\textasteriskcentered &\textasteriskcentered &\textasteriskcentered &\textasteriskcentered &(12,1)_1 &(12,1)_2 &(12,1)_3 \\
 			(2,1)&(12,1)_1 &(12,1)_2 &(12,1)_3 &(12,1)_4 &(12,1)_5 &\textasteriskcentered &\textasteriskcentered &\textasteriskcentered \\
 			(3,1)&(13,1)_1 &(13,1)_2 &(13,1)_3 &(13,1)_4 &(13,1)_5 &(23,1)_1 &(23,1)_2 &(23,1)_3 \\
 			(1,2) &\textasteriskcentered &\textasteriskcentered &\textasteriskcentered &\textasteriskcentered &\textasteriskcentered &(12,2)_1 &(12,2)_2 &(12,2)_3 \\
 			(2,2)&(12,2)_1 &(12,2)_2 &(12,2)_3 &(12,2)_4 &(12,2)_5 &\textasteriskcentered &\textasteriskcentered &\textasteriskcentered \\
 			(3,2)&(13,2)_1 &(13,2)_2 &(13,2)_3 &(13,2)_4 &(13,2)_5 &(23,2)_1 &(23,2)_2 &(23,2)_3 
 		\end{bNiceMatrix},\\&B_{1,2}=\begin{bNiceMatrix}[first-col,first-row]{} &27&28 &34&35&36&37&38\\
 			(12,1) &\textasteriskcentered &\textasteriskcentered&(123,1)_1 &(123,1)_2 &(123,1)_3 &(123,1)_4 &(123,1)_5\\
 			(13,1) &(123,1)_4 &(123,1)_5&\textasteriskcentered &\textasteriskcentered &\textasteriskcentered &\textasteriskcentered &\textasteriskcentered\\
 			(23,1)&\textasteriskcentered &\textasteriskcentered &\textasteriskcentered &\textasteriskcentered &\textasteriskcentered &\textasteriskcentered &\textasteriskcentered\\
 			(1,1) &(12,1)_4 &(12,1)_5&(13,1)_1 &(13,1)_2 &(13,1)_3 &(13,1)_4 &(13,1)_5\\
 			(2,1) &\textasteriskcentered &\textasteriskcentered&(23,1)_1 &(23,1)_2 &(23,1)_3 &(23,1)_4 &(23,1)_5\\
 			(3,1) &(23,1)_4 &(23,1)_5 &\textasteriskcentered &\textasteriskcentered &\textasteriskcentered &\textasteriskcentered &\textasteriskcentered \\
 			(1,2) &(12,2)_4 &(12,2)_5&(13,2)_1 &(13,2)_2 &(13,2)_3 &(13,2)_4 &(13,2)_5\\
 			(2,2) &\textasteriskcentered &\textasteriskcentered&(23,2)_1 &(23,2)_2 &(23,2)_3 &(23,2)_4 &(23,2)_5\\
 			(3,2) &(23,2)_4 &(23,2)_5 &\textasteriskcentered &\textasteriskcentered &\textasteriskcentered &\textasteriskcentered &\textasteriskcentered 
 		\end{bNiceMatrix}
 	\end{align*}\endgroup
 	\caption{$B_{1,1}$ and $B_{1,2}$}
 	\label{B1}
 \end{figure*}
 	\begingroup
 \setlength\arraycolsep{3pt}
 \begin{align*}
 	B_2=\begin{bNiceMatrix}[first-col,first-row]
 		{} &12 &13 &23\\
 		(1,1) &\textasteriskcentered &\textasteriskcentered &(123,1)\\
 		(2,1) &\textasteriskcentered &(123,1) &\textasteriskcentered\\
 		(3,1) &(123,1) &\textasteriskcentered &\textasteriskcentered\\
 	\end{bNiceMatrix}
 \end{align*}\endgroup
 Observe that $F^{\prime}_2=3,Z^{\prime}_2=2,F^{\prime}_1=9$ and $Z^{\prime}_1=4$. Now, consider a subset $I\subset X=[8]$ of size $t=3$, say $I=\{2,4,6\}$. Thus, we have $\tau_1(I)=\{12,23,25,27,28,14,34,45,47,48,16,36,56,67,68\}$ and $\tau_2(I)=\{24,26,46\}$ which are isomorphic to $\mathcal{C}_1=\{14,15,16,17,18,24,25,26,27,28,34,35,36,37,38\}$ and $\mathcal{C}_2=\{12,13,23\}$, respectively, under the relabeling of the elements of $I=\{2,4,6\}$ to $[3]$.

 Let $A^I_H(i)$ denote the $i^{th}$ block in $\mathcal{A}$ containing all the elements in $H$ but none of the elements in $I\setminus H$, with the blocks arranged lexicographically. From Corollary \ref{lambdats}, there are $\lambda^t_s$ such blocks for each $H\subset I$, $|I|=t$ and $|H|=s<t$. For $I=\{2,4,6\}$ and for every $H\in\binom{I}{s}$, we have $\zeta_{j}(I)=\bigcup\limits_{s=1}^{t-j} \{A^I_H(i):H\in\binom{I}{s},i\in[a_{s,j}]\}$. Hence, for $s=1$ and $H\in\binom{I}{1}$, we have $A^I_{2}=\{1278,2358\},A^I_{4}=\{1458,3478\}$ and $A^I_{6}=\{1368,5678\}.$ Further, for $s=2$ and $H\in\binom{I}{2}$, we have $A^I_{24}=\{1234,2457\},A^I_{26}=\{1256,2367\}$ and $A^I_{46}=\{1467,3456\}$. Therefore, we have $\zeta_1(I)=\{1278,2358,1458,3478,1368,5678,1234,1256,1467\}$ and $\zeta_2(I)=\{1278,1458,1368\}$. Under the relabeling of the elements of $I=\{2,4,6\}$ to $[3]$, the sets $\zeta_1(I)$ and $\zeta_2(I)$ are isomorphic to $\mathcal{R}_1=\{(1,1),(1,2),(2,1),(2,2),(3,1),(3,2),(12,1),(13,1),\\(23,1)\}$ and $\mathcal{R}_2=\{(1,1),(2,1),(3,1)\},$ respectively. 
 
 Since the elements of the set $\zeta_1(I)$ and $\zeta_2(I)$ are isomorphic to $\mathcal{R}_1$ and $\mathcal{R}_2$, respectively, and the elements of $\tau_1(I)$ and $\tau_2(I)$ are isomorphic to $\mathcal{C}_1$ and $\mathcal{C}_2$, respectively, the generalized HpPDA condition $[P]_{\zeta_j(I)\times\tau_j(I)}\overset{\textasteriskcentered}{=}B_j,\;\forall j\in[r]$ is satisfied.
 
 We now explain how the hotplug coded caching scheme is constructed.
 
 \textit{Placement Phase}: Recall that $Z$ is the total number of coded subfiles of each file available to user $U$ through the caches it accesses. From Theorem \ref{lambdas}, in a $t$-$(v,k,\lambda)$ design, any subset $Y\subset X$ with $|Y|=s$ is contained in exactly $\lambda_s$ blocks. Since $P(A,U)=\textasteriskcentered$ if $U\cap A\not=\emptyset$, the number of rows in which a column $U$ contains a $\textasteriskcentered$ symbol is given as $Z=\sum\limits_{i=1}^{r}(-1)^{i+1}\binom{r}{i}\lambda_i$. For the given system, there are $Z=\sum\limits_{i=1}^{2}(-1)^{i+1}\binom{2}{i}\lambda_i=2\lambda_1-\lambda_2=14-3=11$. 
 
 Moreover, $F^{\prime}_r=F^{\prime}_2=3$ and $Z^{\prime}_r=Z^{\prime}_2=2$. Hence, each file $W_n$ is partitioned into $F^{\prime}_2-Z^{\prime}_2+Z=3-2+11=12$ subfiles as $W_{n}=\{W_{n,i}:i\in[12]\},\;\forall n\in[N].$

Since there are $b=\lambda_0=14$ blocks, each subfile of every file is coded into $14$ coded subfiles using a $[14,12]$ MDS code, indexed as $W^c_{n,A}:A\in\mathcal{A},\;\forall n\in[N]$. These coded subfiles are placed into the caches using the array $P_c$ as $Z_i=\{W^c_{n,A}:i\in A\}$. Since each element occurs in exactly $\lambda_1=7$ blocks, each cache stores $7$ coded subfiles, giving $\frac{M}{N}=\frac{\lambda_1}{F^{\prime}_2-Z^{\prime}_2+Z}=\frac{7}{12}$.
 
 \textit{Delivery Phase}: Using the PDA $B_2$ and the isomorphism between $(\tau_2(I),\zeta_2(I))$ and $(\mathcal{C}_2,\mathcal{R}_2)$, the server constructs the array $\overline{P}$ as given in Line \ref{deliveryP} of Algorithm \ref{Algo1}. Using this array, the server forms the coded multicast transmission: \begin{enumerate}\item $X_1=W^c_{d_{24},1368}\oplus W^c_{d_{26},1458}\oplus W^c_{d_{46},1278}$. \end{enumerate}
 
 Observe that the user $24$ has access to coded subfiles $W^c_{n,1458}$ and $W^c_{n,1278}$ and is able to decode the transmission. The same logic holds for other users in the transmission. 
 
 Thus, each user in the set $\{24,26,46\}$ receives one coded subfile from the transmission and already has access to $11$ subfiles from the caches it connects to, obtaining $12$ subfiles. Using the property of the $[14,12]$ MDS code, each of these users will be able to reconstruct its desired files.
 
 Next, using the PDA $B_1$ and the isomorphism between $(\tau_1(I),\zeta_1(I))$ and $(\mathcal{C}_1,\mathcal{R}_1)$, the server again constructs the array $\overline{P}$ as given in Line \ref{deliveryP} of Algorithm \ref{Algo1}. The resulting multicast transmissions are: 
 \begin{enumerate}\item $X_2=W^c_{d_{12},1467}\oplus W^c_{d_{14},1256}\oplus W^c_{d_{16},1234}$
 	\item $X_3=W^c_{d_{23},1467}\oplus W^c_{d_{34},1256}\oplus W^c_{d_{36},1234}$
 	\item $X_4=W^c_{d_{25},1467}\oplus W^c_{d_{45},1256}\oplus W^c_{d_{56},1234}$
 	\item $X_5=W^c_{d_{27},1467}\oplus W^c_{d_{47},1256}\oplus W^c_{d_{67},1234}$
 	\item $X_6=W^c_{d_{28},1467}\oplus W^c_{d_{48},1256}\oplus W^c_{d_{68},1234}$
 	\item $X_7=W^c_{d_{12},1458}\oplus W^c_{d_{14},1278}$
 	\item $X_8=W^c_{d_{12},3478}\oplus W^c_{d_{14},2358}$
 	\item $X_9=W^c_{d_{23},1458}\oplus W^c_{d_{34},1278}$
 	\item $X_{10}=W^c_{d_{23},3478}\oplus W^c_{d_{34},2358}$
 	\item $X_{11}=W^c_{d_{25},1458}\oplus W^c_{d_{45},1278}$
 	\item $X_{12}=W^c_{d_{25},3478}\oplus W^c_{d_{45},2358}$
 	\item $X_{13}=W^c_{d_{27},1458}\oplus W^c_{d_{47},1278}$
 	\item $X_{14}=W^c_{d_{27},3478}\oplus W^c_{d_{47},2358}$
 	\item $X_{15}=W^c_{d_{28},1458}\oplus W^c_{d_{48},1278}$
 	\item $X_{16}=W^c_{d_{28},3478}\oplus W^c_{d_{48},2358}$
 	\item $X_{17}=W^c_{d_{12},1368}\oplus W^c_{d_{16},1278}$
 	\item $X_{18}=W^c_{d_{12},5678}\oplus W^c_{d_{16},2358}$
 	\item $X_{19}=W^c_{d_{23},1368}\oplus W^c_{d_{36},1278}$
 	\item $X_{20}=W^c_{d_{23},5678}\oplus W^c_{d_{36},2358}$
 	\item $X_{21}=W^c_{d_{25},1368}\oplus W^c_{d_{56},1278}$
 	\item $X_{22}=W^c_{d_{25},5678}\oplus W^c_{d_{56},2358}$
 	\item $X_{23}=W^c_{d_{27},1368}\oplus W^c_{d_{67},1278}$
 	\item $X_{24}=W^c_{d_{27},5678}\oplus W^c_{d_{67},2358}$
 	\item $X_{25}=W^c_{d_{28},1368}\oplus W^c_{d_{68},1278}$
 	\item $X_{26}=W^c_{d_{28},5678}\oplus W^c_{d_{68},2358}$
 	\item $X_{27}=W^c_{d_{14},1368}\oplus W^c_{d_{16},1458}$
 	\item $X_{28}=W^c_{d_{14},5678}\oplus W^c_{d_{16},3478}$
 	\item $X_{29}=W^c_{d_{34},1368}\oplus W^c_{d_{36},1458}$
 	\item $X_{30}=W^c_{d_{34},5678}\oplus W^c_{d_{36},3478}$
 	\item $X_{31}=W^c_{d_{45},1368}\oplus W^c_{d_{56},1458}$
 	\item $X_{32}=W^c_{d_{45},5678}\oplus W^c_{d_{56},3478}$
 	\item $X_{33}=W^c_{d_{47},1368}\oplus W^c_{d_{67},1458}$
 	\item $X_{34}=W^c_{d_{47},5678}\oplus W^c_{d_{67},3478}$
 	\item $X_{35}=W^c_{d_{48},1368}\oplus W^c_{d_{68},1458}$ 
 	\item $X_{36}=W^c_{d_{48},5678}\oplus W^c_{d_{68},3478}$
 \end{enumerate}
 
 Consider the transmission $X_2=W^c_{d_{12},1467}\oplus W^c_{d_{14},1256}\oplus W^c_{d_{16},1234}$. Observe that user $12$ can only access the cache contents of cache $2$. However, the coded subfiles $W^c_{n,1256}$ and $W^c_{n,1234}$ are in cache $2$. Hence, user $12$ is able to decode the transmission. Similarly, for the transmission $X_7=W^c_{d_{12},1458}\oplus W^c_{d_{14},1278}$, user $14$ can access the cache $4$, which stores the coded subfile $W^c_{n,1458}$, enabling user $14$ to decode the transmission. A similar argument holds for all other involved users.
 
 Each user in this group receives $5$ coded subfiles from the above transmissions and has access to $7$ coded subfiles via the online cache it connects to, resulting in $12$ coded subfiles. Thus, each user is able to reconstruct its desired file.

 \textit{Rate}: Since the PDA $B_2$ generates $1$ transmission and the PDA $B_1$ generates $5\times(1+3\times2)=35$ transmissions, the total number of server transmissions is $36$. Further, each file is partitioned into $12$ subfiles, leading to the rate $R=\frac{36}{12}=3$.
 \end{example}We now prove that the array $B_j$ is a PDA, $j\in[r]$.
 \begin{thm}
 	\label{pdaproof}
 	For $j\in[r],s\in[t-j]$ and $0\leq a_{s,j}\leq \lambda^t_s$, the array $B_j$ defined in \eqref{Bjdefn} is a $[K^{\prime}_j,F^{\prime}_j,Z^{\prime}_j,S_j]$ PDA where 
 	\begin{align*}
 		&K^{\prime}_j=\binom{t}{j}\binom{v-t}{r-j}, F^{\prime}_j=\sum\limits_{s=1}^{t-j} a_{s,j}\binom{t}{s},\\&Z^{\prime}_j=\sum\limits_{s=1}^{t-j} a_{s,j}\left(\binom{t}{s}-\binom{t-j}{s}\right),\\& S_j=\sum\limits_{s=1}^{t-j} a_{s,j}\binom{t}{s+j}\binom{v-t}{r-j}.
 	\end{align*}
 	\end{thm}
 	\begin{proof}
 		The rows of $B_j$ are indexed by the elements of the set $\mathcal{R}_j$ and its columns by $\mathcal{C}_j$. Since $|\mathcal{R}_j|=\sum\limits_{s=1}^{t-j}a_{s,j}\binom{t}{s}$ and $|\mathcal{C}_j|=\binom{t}{j}\binom{v-t}{r-j}$, we have $K^{\prime}_j=\binom{t}{j}\binom{v-t}{r-j}$ and $F^{\prime}_j=\sum\limits_{s=1}^{t-j}a_{s,j}\binom{t}{s}$. 
 		
 		By construction, we have $B_j((Y,i),U)=\textasteriskcentered$ if $\{U\cap[t]\}\cap Y\not=\emptyset$. For a given $s\in[t-j]$ and $|U\cap[t]|=j$, the number of subsets $Y$ of cardinality $s$ satisfying $Y\cap\{U\cap[t]\}=\emptyset$ is $\binom{t-j}{s}$. Hence, there will be $\sum\limits_{s=1}^{t-j} a_{s,j}\binom{t-j}{s}$ non-star entries in each column. Since there are $F^{\prime}_j$ rows, we have $Z^{\prime}_j=\sum\limits_{s=1}^{t-j}a_{s,j}\left( \binom{t}{s}-\binom{t-j}{s}\right)$. 
 		
 		Finally, we have $B_j((Y,i),U)=(Y\cup\{U\cap[t]\},i)_{n_{Y\cup\{U\cap[t]\}}}$ if $Y\cap\{U\cap[t]\}=\emptyset$. Since $|U\cap[t]|=j$ and $|Y|=s$, the cardinality of $Y\cup\{U\cap[t]\}$ is $s+j$. Thus, for a fixed $s\in[t-j]$, there are $\binom{t}{s+j}$ such subsets. Moreover, for each such subset of cardinality $s+j$, there are $\binom{v-t}{r-j}$ choices of $U$ that yield the same union. Finally, as $i\in[a_{s,j}]$, each union will occur $a_{s,j}$ times, leading to $S_j=\sum\limits_{s=1}^{t-j} a_{s,j}\binom{t}{s+j}\binom{v-t}{r-j}$ non-star entries.
 		
 		We now prove Condition C3. a. as given in Definition \ref{pdadefn}. 
 		
 		Consider two distinct integer entries $B((Y_1,i),U_1)$ and $B((Y_2,i),U_2)$ such that $B((Y_1,i),U_1)=B((Y_2,i),U_2)=(Y_1\cup\{U_1\cap[t]\},i)_{n_{Y_1\cup\{U_1\cap[t]\}}}=(Y_2\cup\{U_2\cap[t]\},i)_{n_{Y_2\cup\{U_2\cap[t]\}}}$. This implies that \begin{align}\label{pdaproofconditionc3a}Y_1\cup\{U_1\cap[t]\}=Y_2\cup\{U_2\cap[t]\}.\end{align}Because of this, if $Y_1=Y_2$, then $U_1=U_2$ or if $U_1=U_2$, then $Y_1=Y_2$. Both cases contradict the assumption that the two entries are distinct. Thus, we have $(Y_1,i)\not=(Y_2,i)$ and $U_1\not=U_2$, proving the condition.
 		
 		We now prove the Condition C3. b. given in Definition \ref{pdadefn}. 
 		
 		Consider an element $u_1\in\{U_1\cap[t]\}\setminus\{U_2\cap[t]\}$. We know that $u_1\not\in Y_1$ and by construction, $u_1\not\in\{U_2\cap[t]\}$. From Equation \eqref{pdaproofconditionc3a}, it follows that $u_1\in Y_2$, which implies that $Y_2\cap\{U_1\cap[t]\}\not=\emptyset$, leading to $B((Y_2,i),U_1)=\textasteriskcentered$. Similarly, it can be shown that $B((Y_1,i),U_2)=\textasteriskcentered$, proving the condition.
 	\end{proof}
 	Next, we present the proof that there exists a pair $(\tau_j(I),\zeta_j(I))$ such that $ 		[P]_{\zeta_j(I)\times\tau_j(I)}\overset{\textasteriskcentered}{=}B_j$.
 	\begin{thm}
 		For $j\in[r],s\in[t-j]$, and $0\leq a_{s,j}\leq \lambda^t_s$, the tuple $(P,\{B_j\}_{j=1}^{r})$ forms a $(C,C^{\prime},r,F,Z_c,Z,\{F^{\prime}_j,Z^{\prime}_j,S_j\}_{j=1}^{r})$-generalized HpPDA, with the parameters $\{F^{\prime}_j,Z^{\prime}_j,S_j\}_{j=1}^{r}$ defined in Theorem \ref{pdaproof} and $C=v,C^{\prime}=t, F=b, Z_c=\lambda_1,$ and $Z=\sum\limits_{i=1}^{r}(-1)^{i+1}\binom{r}{i}\lambda_i$.
 	\end{thm}
 	\begin{proof}
 		From Theorem \ref{pdaproof}, we know that $F^{\prime}_j=\sum\limits_{s=1}^{t-j} a_{s,j}\binom{t}{s},Z^{\prime}_j=\sum\limits_{s=1}^{t-j} a_{s,j}\left(\binom{t}{s}-\binom{t-j}{s}\right)$, and $S_j=\sum\limits_{s=1}^{t-j} a_{s,j}\binom{t}{s+j}\binom{v-t}{r-j}.$ 
 		
 		The cache placement array $P_c$, defined in \eqref{pcdefn}, is a $b\times v$ array with rows indexed by $A\in\mathcal{A}$ and columns by elements of $X$, where $P_c(A,i)=\textasteriskcentered$ if $i\in A$. From Theorem \ref{lambdas}, each element occurs in exactly $\lambda_1$ blocks and hence each column of $P_c$ has $Z_c=\lambda_1$ stars.
 		
 		The array $P$, defined in \eqref{pdefn}, is an array of dimension $b\times\binom{v}{r}$ containing $\textasteriskcentered$ and nulls, with rows being indexed by the blocks in $\mathcal{A}$ and columns by $U\in\binom{X}{r}$. From Theorem \ref{lambdas}, a set $Y,|Y|=s<t$ occurs in $\lambda_s$ blocks. Using this property and by construction of $P$, we have $Z=\sum\limits_{i=1}^{r}(-1)^{i+1}\binom{r}{i}\lambda_i$. For a user connecting to $j$ online caches, it would have access to the subfiles stored in those caches. Thus, such a user will have access to $Z_j=\sum\limits_{i=1}^{j}(-1)^{i+1}\binom{j}{i}\lambda_i$ subfiles.
 		
 		We now prove \eqref{star=condition} given in Definition \ref{CHpPDAdefn}. Fix a subset $I\subset X,|I|=t$. For any $H\subset I,|H|=s<t$, Corollary \ref{lambdats} states that there are exactly $\lambda^t_s$ blocks containing $H$ and none of the elements of $I\setminus H$. 
 		
 		We denote these blocks by the set $A^I_H=\{A^I_H(i):i\in[\lambda^t_s]\}$, where $A^I_H(i)$ is the $i^{th}$ block in this set under lexicographical ordering. 
 		
 		For $I$ and $H$ defined above and a fixed $j\in[r]$, define the set $\zeta_j(I)=\bigcup\limits_{s=1}^{t-j} \{A^I_H(i):i\in[a_{s,j}],H\in\binom{I}{s}\}$, where $0\leq a_{s,j}\leq \lambda^t_s$. This set has $a_{s,j}$ blocks for each $H$, leading to $|\zeta_j(I)|=\sum\limits_{s=1}^{t-j}a_{s,j}\binom{t}{s}$. We write the set $\zeta_j(I)$ as $\zeta_j(I)=\bigcup\limits_{s=1}^{t-j}\zeta_{j,s}(I)$, where $\zeta_{j,s}(I)=\{A^I_H(i):i\in[a_{s,j}],H\in\binom{I}{s}\}$. 
 		
 		Finally, consider the subarrays $[P]_{\zeta_{j,s}(I),\tau_j(I)}$ and $[B_j]_{\mathcal{R}_{j,s},\mathcal{C}_j}$, where $\mathcal{R}_{j,s}=\{(Y,i):Y\in\binom{[t]}{s},i\in[a_{s,j}]\}$. Both arrays have dimension $a_{s,j}\binom{t}{s}\times\binom{t}{j}\binom{C-t}{r-j}$. 
 		
 		We first show isomorphism between the column indexing of $[P]_{\zeta_{j,s}(I),\tau_j(I)}$ and $[B_j]_{\mathcal{R}_{j,s},\mathcal{C}_j}$. Observe that $\tau_j(I)$ consists of all users with $|U\cap I|=j$ and $\mathcal{C}_{j}$ contains all $r$-subsets of $X$ intersecting with $[t]$ in $j$ elements. Thus, a bijection $f_1:I\mapsto[t], f_1(I(i))=i$ exists, where $I(i)$ is the $i^{th}$ element of $I$, ordered lexicographically.
 		
 		Now, we show the isomorphism between $\mathcal{R}_{j,s}$ and $\zeta_{j,s}(I)$. The elements of $\mathcal{R}_{j,s}$ are $(Y,i), Y\in\binom{[t]}{s},i\in[a_{s,j}]$ whereas $\zeta_{j,s}(I)$ is a set of $a_{s,j}$ blocks $A^I_H(i):H\in\binom{I}{s},i\in[a_{s,j}]$. Hence, both sets contain $a_{s,j}$ elements of all possible $s$ subsets of $[t]$ and $I$, respectively. Thus, for each block $A^I_H(i)$, there is a unique $(Y,i)$ for some $Y\in\binom{[t]}{s},i\in[a_{s,j}]$. Thus, we define the bijection $f_2:A^I_H(i)\mapsto(f_1(H),i)$. The above is true for all $s\in[t-j]$.
 		
 		By construction, as $P(A,U)=\textasteriskcentered$ if $A\cap U\not=\emptyset$ and $B_j((Y,i),U)=\textasteriskcentered$ if $Y\cap\{U\cap[t]\}\not=\emptyset$, using the above isomorphism, both conditions are equivalent.
 		
 		Therefore, the condition $[P]_{\zeta_{j}(I)\times\tau_j(I)}\overset{\textasteriskcentered}{=}B_j$ is proved.
 	\end{proof}Using the generalized HpPDA constructed above, we now characterize the achievable rate-memory trade-off of the proposed combinatorial multi-access hotplug coded caching scheme based on $t$-designs in the following corollary:
 	\begin{cor}
	 		Consider a $t$-$(v, k, \lambda)$ design and the corresponding $(C, C^{\prime}, r, N)$ combinatorial multi-access hotplug coded caching model with $C = v$ and $C^{\prime} = t$. For each $j \in [r]$ and $s \in [t - j]$, let $a_{s,j}\in[\lambda^t_s]$ be such that $Y_j\leq Y_{j-1},\;\forall j\in[r]$, where $Y_j=Z_j+F^{\prime}_j-Z^{\prime}_j$. Then the generalized HpPDA describes a combinatorial multi-access hotplug coded caching scheme with cache memory $M=\frac{N\lambda_1}{\sum\limits_{i=1}^{r}(-1)^{i+1}\binom{r}{i}\lambda_i + \sum\limits_{s=1}^{t-r} a_{s,r}\binom{t-r}{s}}$ and rate $R=\frac{\sum\limits_{j=1}^{r}\sum\limits_{s=1}^{t-j}a_{s,j}\binom{t}{s+j}\binom{v-t}{r-j}}{\sum\limits_{i=1}^{r}(-1)^{i+1}\binom{r}{i}\lambda_i + \sum\limits_{s=1}^{t-r} a_{s,r}\binom{t-r}{s}}$.
 	\end{cor}
We now prove that under the condition that $a_{s,j}=\lambda^t_s$, for all $s\in[t-j]$ and $j\in[r]$, every active user has access to the same number of coded subfiles after the server finishes transmissions. 

This choice corresponds to the maximum number of coded subfiles that can be delivered to users connecting to $j$ online caches, since $F^{\prime}_j-Z^{\prime}_j$ is increasing in the parameters $a_{s,j}$. By choosing smaller values of $a_{s,r}$, and appropriately selecting the parameters $a_{s,j},j\in[r-1]$, it is possible to ensure that users connecting to fewer online caches have access to at least as many coded subfiles as those connecting to more online caches. 

Hence, the proposed framework admits multiple valid choices of the parameters $a_{s,j}$ that preserve correctness.
\begin{claim}
	\label{as=lambdatsclaim}
 		Consider a $t$-$(v,k,\lambda)$ design. For the proposed construction, let $Y_j=Z_j+F^{\prime}_j-Z^{\prime}_j,\;\forall j\in[r]$. Then, under the condition $a_{s,j}=\lambda^t_s,s\in[t-j],j\in[r]$, $Y_j$ is invariant with respect to $j$, that is, $Y_j=Y_{j-1},\forall j\in[r]$.
 	\end{claim}
 	\begin{proof}
 		We first calculate $Y_j-Y_{j-1}$ and then show that $Y_j-Y_{j-1}=0$. By definition,\begin{align*}
 			&Y_j-Y_{j-1}=Z_j+F^{\prime}_j-Z^{\prime}_j-(Z_{j-1}+F^{\prime}_{j-1}-Z^{\prime}_{j-1})\\
 			=&\sum\limits_{i=1}^{j}(-1)^{i+1}\binom{j}{i}\lambda_i+\sum\limits_{s=1}^{t-j}\lambda^t_s\binom{t-j}{s}-\\&\left(\sum\limits_{i=1}^{j-1}(-1)^{i+1}\binom{j-1}{i}\lambda_i+\sum\limits_{s=1}^{t-j+1}\lambda^t_s\binom{t-j+1}{s}\right)\\
 			=&\sum\limits_{i=1}^{j-1}(-1)^{i+1}\lambda_i\left(\binom{j}{i}-\binom{j-1}{i}\right) + (-1)^{j+1}\lambda_j +\\&\sum\limits_{s=1}^{t-j}\lambda^t_s \left(\binom{t-j}{s}-\binom{t-j+1}{s}\right) - \lambda^t_{t-j+1}
 			\end{align*}
 			\begin{align*}
 			=&\sum\limits_{i=1}^{j-1}(-1)^{i+1}\lambda_i\binom{j-1}{i-1} + (-1)^{j+1}\lambda_j -\\&\sum\limits_{s=1}^{t-j}\lambda^t_s\binom{t-j}{s-1}-\lambda^t_{t-j+1}\\
 			=&\sum\limits_{i=1}^{j-1}(-1)^{i+1}\lambda_i\binom{j-1}{i-1} + (-1)^{j+1}\binom{j-1}{j-1}\lambda_j -\\&\sum\limits_{s=1}^{t-j}\lambda^t_s\binom{t-j}{s-1}-\lambda^t_{t-j+1}\binom{t-j}{t-j}\\
 			=&\sum\limits_{i=1}^{j}(-1)^{i+1}\lambda_i\binom{j-1}{i-1} -\sum\limits_{s=1}^{t-j+1}\lambda^t_s\binom{t-j}{s-1}
 		\end{align*}
 		We use the Pascal's Identity $\binom{n+1}{k}=\binom{n}{k}+\binom{n}{k-1}$ to simplify the above expressions. We now prove that $\sum\limits_{i=1}^{j}(-1)^{i+1}\lambda_i\binom{j-1}{i-1} =\sum\limits_{s=1}^{t-j+1}\lambda^t_s\binom{t-j}{s-1}$ to show that $Y_j=Y_{j-1}$.
 		
 		Fix a set $T\subset X,|T|=t$, a subset $J\subset T,|J|=j$ and an element $x\in J$. Define $\mathcal{F}=\{A:\mathcal{A},x\in A,A\cap\{J\setminus\{x\}\}=\emptyset\}$. We calculate $|\mathcal{F}|$ in two different ways.
 		
 		\textit{Step 1}: For each $y\in J\setminus \{x\}$, define the set $B_y=\{A\in\mathcal{A}:\{x,y\}\in A\}.$ Then $\mathcal{F}=\{A:A\in\mathcal{A},x\in A\}\setminus\bigcup\limits_{y\in J\setminus\{x\}} B_y$. Using the principle of inclusion-exclusion, $|\mathcal{F}|=|\{A:A\in\mathcal{A},x\in A\}\setminus\bigcup\limits_{y\in J\setminus\{x\}} B_y|=\sum\limits_{\mathcal{I}\subset J\setminus\{x\}} (-1)^{|\mathcal{I}|}\bigcap\limits_{y\in \mathcal{I}} B_y|$. Hence, $|\mathcal{F}|=\sum\limits_{\mathcal{I}\subset J\setminus\{x\}} (-1)^{|\mathcal{I}|} |\{A\in\mathcal{A}:\{x\}\cup \mathcal{I}\subset A\}|.$ Let $\mathcal{I}\subset J\setminus\{x\}$ with $|\mathcal{I}|=i-1$. Thus, the set $\{x\}\cup\mathcal{I}$ is a set of cardinality $i$ and occurs in $\lambda_i$ blocks. As $|J\setminus\{x\}|=j-1$, the number of such sets $\mathcal{I}$ is $\binom{j-1}{i-1}$, which leads to $|\mathcal{F}|=\sum\limits_{i=1}^{j} (-1)^{i-1} \binom{j-1}{i-1}\lambda_i=\sum\limits_{i=1}^{j}(-1)^{i+1} \binom{j-1}{i-1}\lambda_i$.
 		
 		\textit{Step 2}: Consider a block $A\in\mathcal{F}$. Since $x\in A$, $A\cap\{J\setminus\{x\}\}=\emptyset$ and $J\subset T$, it follows that $A\cap T\subset \{x\}\cup\{T\setminus J\}$. Let $S\subset T\setminus J$ be a set such that $A\cap T=\{x\}\cup S$. Let $|A\cap T|=s$ which gives $|S|=s-1$. Since $|T\setminus J|=t-j$ and $S\subset T\setminus J$, we have $0\leq |S|\leq t-j$ which leads to $1\leq s\leq t-j+1$. For an $s$ such that $1\leq s\leq t-j+1$, there are $\binom{t-j}{s-1}$ choices of the set $S$. For each $S$, we have $A\cap T=\{x\}\cup S$. Thus, the block $A$ contains all the points in $\{x\}\cup S$ but none of the points in $T\setminus \{\{x\}\cup S\}$. From the property of the design, there are $\lambda^t_s$ such blocks. Hence, we have $|\mathcal{F}|=\sum\limits_{s=1}^{t-j+1} \lambda^t_s\binom{t-j}{s-1}$.
 		
 		From \textit{Step 1} and \textit{Step 2}, we have $\sum\limits_{i=1}^{j}(-1)^{i+1} \binom{j-1}{i-1}\lambda_i=\sum\limits_{s=1}^{t-j+1} \lambda^t_s\binom{t-j}{s-1}$, leading to $Y_j=Y_{j-1}$.
 	\end{proof}Notice that $Y_j$ denotes the number of coded subfiles that a user connecting to exactly $j$ online caches has access to after the server completes all transmissions. Under the condition $a_{s,j}=\lambda_s^t$ for all $s\in[t-j]$ and $j\in[r]$, we have $Y_j=Y_{j-1}$ for all $j\in[r]$, and hence $Y_r = Y_{r-1} = \cdots = Y_1$.
 	
 	In general, we may choose parameters $a_{s,j}<\lambda_s^t$ such that every active user has access to at least $Y_r$ coded subfiles after the transmission phase. Under such choices, the relation $Y_j \leq Y_{j-1}$ holds for all $j\in[r]$. Consequently, some users can obtain strictly more coded subfiles than the minimum required for decoding.
 	
 	Since all users require only $Y_r$ coded subfiles to reconstruct their requested files, the additional coded subfiles received by such users correspond to redundant transmissions.
 
 	 \subsection{Specialization to Existing Schemes}
 	 We now show that the proposed $t$-scheme specializes to the CRR $t$-scheme \cite{CRR} for $r=1$.
 	 \subsubsection{CRR $t$-design scheme \cite{CRR}}
 	 We consider the proposed $t$-scheme in the special case $r=1$. Recall that the array $P$ is defined as \begin{align*}P(A,U)=
 	 	\begin{cases}
 	 		\textasteriskcentered,&\text{ if } A\cap U\not=\emptyset\\
 	 		null,&\text{ if } A\cap U=\emptyset
 	 	\end{cases}
 	 \end{align*}for $U\in\binom{X}{r}$ and $A\in\mathcal{A}$. 
 	 
 	 When $r=1$, each $r$-subset indexing the columns reduces to a singleton, hence $U=\{u\}$ for $u\in X$. In this case, the above array reduces to the array $P_c$ defined in \eqref{pcdefn}. Observe that $P_c$ is identical to the array $P$ defined in \cite{CRR}.
 	 
 	 Next, consider the PDA $B_j$, for $j\in[r]$. Since $r=1$, $j$ can take only one value, which is $j=1$, leading to one PDA $B_1$. For $r=1$, the row and column index sets of $B_1$ are given as $\mathcal{R}_1=\bigcup\limits_{s=1}^{t-1} \{(Y,i):Y\in\binom{[t]}{s},i\in[a_{s,1}]\}$ and $\mathcal{C}_1=\{u:u\in X,|u\cap[t]|=1\}=\{u:u\in[t]\}=[t]$. Hence, the PDA $B_1$ is defined as
 	 \begin{align*}&B_1((Y,i),u)=
 	 	\begin{cases}
 	 		\textasteriskcentered,&\text{ if } u\in A\\
 	 		((Y\cup u),i)_{n_{Y\cup u}},&\text{ if } u\not\in A
 	 	\end{cases}
 	 \end{align*}for $(Y,i)\in\mathcal{R}_1$ and $u\in\mathcal{C}_1$. 
 	 
 	 Observe that the PDA $B_1$, together with the sets $\mathcal{R}_1$ and $\mathcal{C}_1$ coincide exactly with the PDA $B$, the row set $\mathcal{R}$ and the column set $[t]$ in \cite{CRR}. 
 	 
 	 Therefore, for $r=1$, the proposed $t$-scheme specializes to the CRR $t$-scheme given in \cite{CRR}.
\section{Numerical Comparison}
\label{numericalcomparisons}
In this section, we compare the proposed combinatorial multi-access $t$-scheme with existing coded caching schemes. Specifically, we compare the proposed scheme with the CRR $t$-scheme and the CRR MT scheme in \cite{CRR} as well as the RR scheme in \cite{RR1}. We also include the MT scheme in \cite{MT}.

In the proposed combinatorial multi-access model, each user connects to an $r$-subset of the caches. As a result, when $C^{\prime}$ caches are online during the delivery phase, the number of active users is $K_o=\sum\limits_{j=1}^{r} \binom{C^{\prime}}{j}\binom{C-C^{\prime}}{r-j}$. In contrast, the earlier hotplug coded caching models considered in \cite{CRR,RR1,MT} assume a dedicated cache structure and therefore support $K_o=C^{\prime}$ active users.

To enable a fair comparison across schemes with different numbers of active users, we normalize the rate by the number of active users $K_o$. Accordingly, in Fig. \ref{numericalcomparisonfig}, the horizontal axis represents the caching fraction $\frac{M}{N}$, while the vertical axis represents the rate per active user.

We consider a combinatorial multi-access hotplug coded caching system with $C=8$ caches such that each user connects to $r=2$ caches and $C^{\prime}=3$ caches are online during the delivery phase, with $N=18$ files. For this configuration, the proposed scheme serves $K_o=18$ active users. For the hotplug schemes in \cite{CRR,RR1,MT}, we use the corresponding dedicated caching setting with $C=8,C^{\prime}=3$ and $N=3$, consistent with their original formulations.

As shown in Fig. \ref{numericalcomparisonfig}, the proposed $t$-scheme achieves a lower rate per active user than the existing schemes over the cache memory range $\frac{M}{N}\in[0.5,0.7]$. This illustrates that, even under hotplug constraints, exploiting increased cache connectivity can lead to improved delivery performance in certain memory regimes.
\begin{figure}
	\includegraphics[width=\columnwidth]{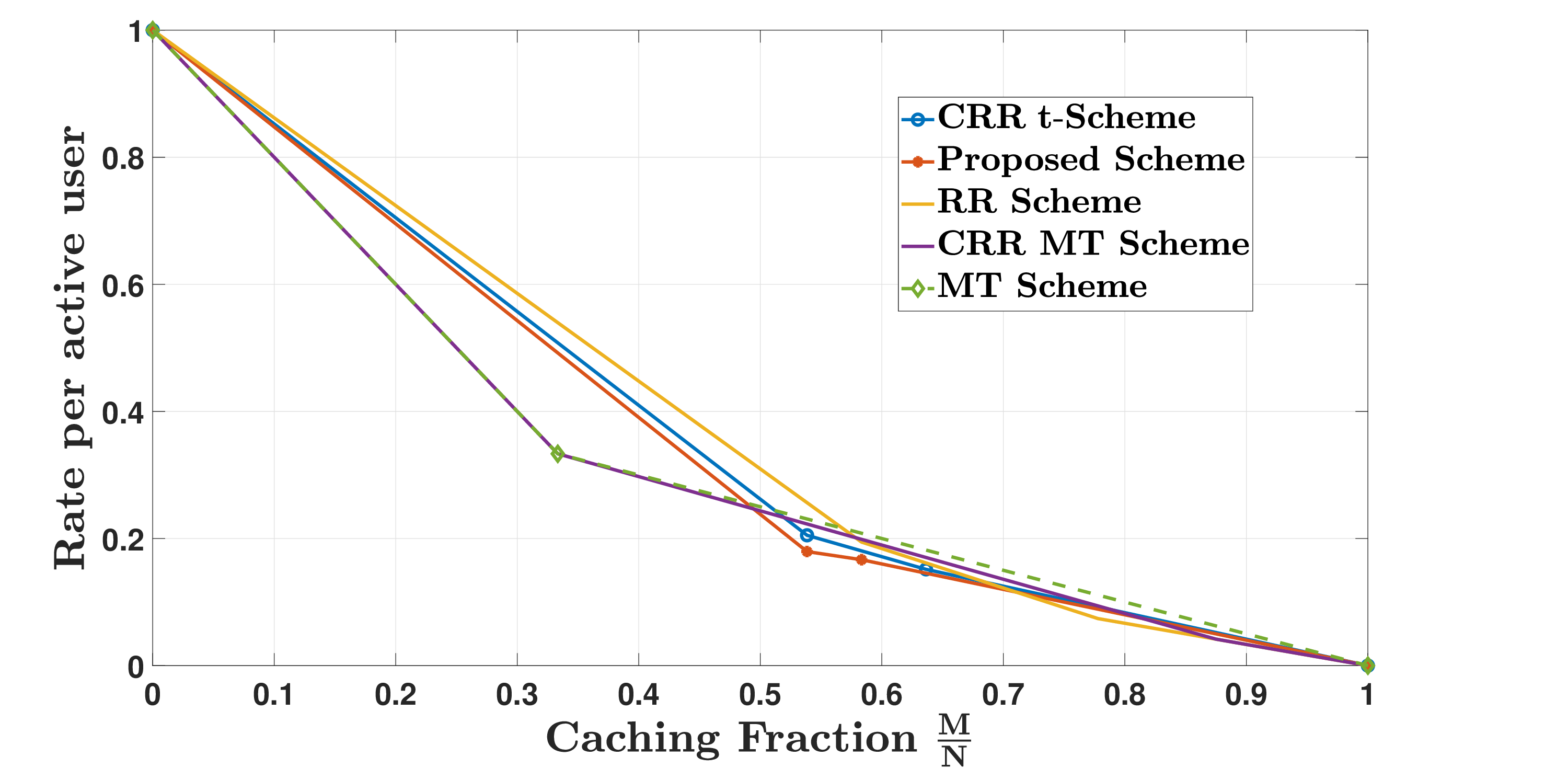}
	\caption{Rate per user \textit{vs.} Caching fraction $\frac{M}{N}$.}
	\label{numericalcomparisonfig}
\end{figure}
\section{Conclusions}
\label{conclusionsandfuturework}
In this paper, we studied coded caching in combinatorial multi-access networks under a hotplug setting, where users access multiple caches and only a subset of caches is available during the delivery phase. We generalized the Hotplug Placement Delivery Array (HpPDA) framework to the combinatorial multi-access model and used this framework to formally describe the placement and delivery procedures.

Using the generalized HpPDA framework, we proposed a $t$-design-based coded caching scheme for combinatorial multi-access hotplug networks. We characterized a class of design parameters under which every active user has access to a sufficient number of coded subfiles to decode its requested file. We further showed that appropriate parameter choices allow the elimination of multicast transmissions that are redundant for decoding, resulting in a family of achievable rate-memory trade-offs with varying subpacketization levels and cache placement structures.

We also showed that the proposed framework specializes to an existing $t$-design-based hotplug coded caching scheme in the single-access setting as a special case. Finally, we provided numerical comparisons demonstrating the performance of the proposed $t$-design scheme relative to existing hotplug coded caching schemes.

The proposed framework is the first to jointly address hotplug operation and combinatorial multi-access architectures, and enables systematic control of redundancy in multicast transmissions through design parameters.
\section*{Acknowledgment}
This work was supported partly by the Science and Engineering Research Board (SERB) of the Department of Science and Technology (DST), Government of India, through J.C Bose National Fellowship to Prof. B. Sundar Rajan.


\begin{thebibliography}{1}
	\bibitem{MAN}
	M. A. Maddah-Ali and U. Niesen, "Fundamental Limits of Caching," in \textit{IEEE Transactions on Information Theory}, vol. 60, no. 5, pp. 2856-2867, May 2014.
	\bibitem{WTP}
	K. Wan, D. Tuninetti and P. Piantanida, "An Index Coding Approach to Caching With Uncoded Cache Placement," in \textit{IEEE Transactions on Information Theory}, vol. 66, no. 3, pp. 1318-1332, March 2020.
	\bibitem{YCTC}
	Q. Yan, M. Cheng, X. Tang and Q. Chen, "On the Placement Delivery Array Design for Centralized Coded Caching Scheme," in \textit{IEEE Transactions on Information Theory}, vol. 63, no. 9, pp. 5821-5833, Sept. 2017.
	\bibitem{MAN1}
	M. A. Maddah-Ali and U. Niesen, "Decentralized Coded Caching Attains Order-Optimal Memory-Rate Tradeoff," in \textit{IEEE/ACM Transactions on Networking}, vol. 23, no. 4, pp. 1029-1040, Aug. 2015.
	\bibitem{KNMAD}
	N. Karamchandani, U. Niesen, M. A. Maddah-Ali and S. N. Diggavi, "Hierarchical Coded Caching," in \textit{IEEE Transactions on Information Theory}, vol. 62, no. 6, pp. 3212-3229, June 2016.
	\bibitem{CFL}
	Z. Chen, P. Fan, and K. B. Letaief, "Fundamental limits of caching: Improved bounds for users with small buffers," \textit{IET Communications}, vol. 10, no. 17, pp. 2315-2318, 2016.
	\bibitem{MKR}
	P. N. Muralidhar, D. Katyal and B. S. Rajan, "Maddah-Ali-Niesen Scheme for Multi-access Coded Caching," \textit{2021 IEEE Information Theory Workshop (ITW)}, Kanazawa, Japan, 2021, pp. 1-6.
	\bibitem{HKD}
	J. Hachem, N. Karamchandani and S. N. Diggavi, "Coded Caching for Multi-level Popularity and Access," in \textit{IEEE Transactions on Information Theory}, vol. 63, no. 5, pp. 3108-3141, May 2017.
	\bibitem{WC}
	K. Wan and G. Caire, "On Coded Caching With Private Demands," in \textit{IEEE Transactions on Information Theory}, vol. 67, no. 1, pp. 358-372, Jan. 2021.
	\bibitem{BE}
	F. Brunero and P. Elia, "Fundamental Limits of Combinatorial Multi-Access Caching," in \textit{IEEE Transactions on Information Theory}, vol. 69, no. 2, pp. 1037-1056, Feb. 2023.
	\bibitem{MT}
	Y. Ma and D. Tuninetti, "On Coded Caching Systems with Offline Users," \textit{2022 IEEE International Symposium on Information Theory (ISIT)}, Espoo, Finland, 2022, pp. 1133-1138.
	\bibitem{RR1}
	C. Rajput and B. Sundar Rajan, "Improved Hotplug Caching Schemes Using PDAs and t-Designs," \textit{arXiv:2311.02856}, 2024.
	\bibitem{CRR}
	M. Chinnapadamala, C. Rajput and B. S. Rajan, "A New Hotplug Coded Caching Scheme Using PDAs," \textit{2024 IEEE Information Theory Workshop (ITW)}, Shenzhen, China, 2024, pp. 478-483.
	\bibitem{S}
	D. R. Stinson, \textit{Combinatorial designs: constructions and analysis}. Springer, 2004, vol. 480.
\end{thebibliography}
\end{document}